\documentclass[twocolumn,aps,prl,%superscriptaddress,
showpacs, amsfonts,amsmath]{revtex4}

\usepackage{graphicx} 

\usepackage{dcolumn}   
\usepackage{bm}        
\usepackage{multirow}
\usepackage{amssymb}
\usepackage{amsfonts}   
\usepackage{amsmath} 
\usepackage{amsthm}
\usepackage{color}
\usepackage{marvosym}
\newcommand{\mathbbm}[1]{\text{\usefont{U}{bbm}{m}{n}#1}} 
\usepackage[colorlinks=true,citecolor=blue,urlcolor=blue]{hyperref}

\newtheorem{theorem}{Theorem}

\newcommand{\bra}[1]{\langle #1|}
\newcommand{\ket}[1]{|#1\rangle}
\newcommand{\braket}[2]{\langle #1|#2\rangle}
\newcommand{\ketbra}[1]{| #1\rangle \langle #1|}

\newcommand{\be}{\begin{equation}}
\newcommand{\ee}{\end{equation}}
\newcommand{\bea}{\begin{eqnarray}}
\newcommand{\eea}{\end{eqnarray}}

\newcommand{\mean}[1]{\ensuremath{\langle{#1}\rangle}}

\newcommand{\BB}{\ensuremath{\mathcal{B}}}

\newcommand{\kommentar}[1]{}

\renewcommand{\vr}{\ensuremath{\varrho}}
\newtheorem{lemma}[theorem]{Lemma}

\newcommand{\forget}[1]{}

\begin{document}

\title{
Extreme violation of local realism in quantum hypergraph states}

\date{\today}

\author{Mariami Gachechiladze}
\author{Costantino Budroni}
\author{Otfried G\"uhne}
\affiliation{Naturwissenschaftlich-Technische Fakult\"at, 
Universit\"at Siegen, Walter-Flex-Str. 3, 57068 Siegen, Germany}

\begin{abstract}
Hypergraph states form a family of multiparticle quantum states 
that generalizes the well-known concept of Greenberger-Horne-Zeilinger 
states, cluster states, and more broadly graph states. We study the 
nonlocal properties of quantum hypergraph states. We demonstrate that 
the correlations in hypergraph states can be used to derive various 
types of nonlocality proofs, including Hardy-type arguments and Bell 
inequalities for genuine multiparticle nonlocality. Moreover, we show 
that hypergraph states allow for an exponentially increasing violation 
of local realism which is robust against loss of particles. Our results 
suggest that certain classes of hypergraph states are novel resources 
for quantum metrology and measurement-based quantum computation.
\end{abstract}

%, 37.10.Ty}
%03.65.Ta: Foundations of quantum mechanics; measurement theory
%03.65.Ud: Entanglement and quantum nonlocality
%(e.g. EPR paradox, Bell's inequalities, GHZ states, etc.)
%42.50.Xa: Optical tests of quantum theory
%37.10.Ty: Ion trapping
\maketitle

{\it Introduction.---} 
Multiparticle entanglement is central for discussions 
about the foundations of quantum mechanics, protocols 
in quantum information processing, and experiments in 
quantum optics. Its characterization has, however, turned 
out to be difficult. One problem hindering the exploration
of multiparticle entanglement is the exponentially increasing
dimension of the Hilbert space. This implies that making statements 
about general quantum states is difficult. So, one has to 
concentrate on families of multiparticle states with an 
easier-to-handle description. In fact, symmetries and other
kinds of simplifications seem to be essential for a state
to be a useful resource. Random states can often be shown
to be highly entangled, but useless for quantum information 
processing \cite{entangleduseful}.

An outstanding class of useful multiparticle quantum states is given by
the family of graph states \cite{hein}, which includes the 
Greenberger-Horne-Zeilinger (GHZ) states and the cluster states as 
prominent examples. Physically, these states have turned out to be relevant
resources for quantum metrology, quantum error correction, or 
measurement-based quantum computation \cite{hein}. Mathematically, these states 
are elegantly given by graphs, which describe the correlations 
and also a possible interaction structure leading to the graph state. 
In addition, graph states can be defined via a so-called stabilizer
formalism: A graph state is the unique eigenstate of a set of commuting
observables, which are local in the sense that they are tensor products
of Pauli measurements. These stabilizer observables are important for
easily computing correlations leading to violations of Bell 
inequalities \cite{Mermin, gthb}, as well as designing simple schemes to characterize
graph states experimentally \cite{gtreview}.

Recently, this family of states has been generalized to hypergraph states 
\cite{Kruszynska2009, Qu2013_encoding, Rossi2013, Otfried, chenlei, lyons}. 
These states have been recognized as special cases of the so-called locally 
maximally entangleable (LME) states \cite{Kruszynska2009}. Mathematically, 
they are described by hypergraphs, a generalization of graphs, where a single hyperedge can connect more than two vertices. They can also  be described by 
a stabilizer formalism, but this time, the stabilizing operators are not 
local. So far, hypergraph states have turned out to play a role for search algorithms in quantum computing \cite{scripta}, quantum fingerprinting protocols \cite{mora}, and they have been shown to be complex enough to serve as witnesses 
in  all QMA problems \cite{qma}. They have recently been 
investigated in condensed matter physics as ground states of spin models with
interesting topological properties \cite{Yoshida, Akimasa}.  In addition, 
equivalence classes and further entanglement properties of hypergraph 
states have been studied \cite{Otfried}.

\begin{figure}
\includegraphics[width=0.45\textwidth ]{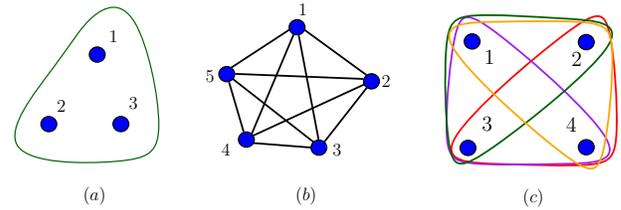}
\caption{Examples of hypergraphs.
(a) A simple hypergraph with three vertices and a single edge 
$e=\{1,2,3\}$ connecting all three vertices. The corresponding
hypergraph state $\ket{H_3}=C_{123}\ket{+}^{\otimes 3}=
(\ket {000}+\ket{001}+\dots +\ket{110}-\ket{111} )/{\sqrt{8}}$
is discussed in detail in the text.
(b) This hypergraph contains only two-edges, so it is an ordinary 
graph. The state corresponding to this fully connected graph is the
five-qubit GHZ state.
(c) The fully connected three-uniform hypergraph represents a state
that can be seen as a generalization of the GHZ state. 
}
\label{fig-hgbild}
\end{figure}

In this paper we show that hypergraph states violate local realism
in an extreme manner, but in a way that is robust against loss of particles.
We demonstrate that this leads to applications of these states in quantum 
metrology and quantum computation. We see that the stabilizer formalism 
describing hypergraph states, despite being nonlocal, can be used 
to derive Hardy-type nonlocality arguments \cite{Hardy92}, Bell 
inequalities for genuine multiparticle entanglement \cite{Svetlichny87},
or a violation of local realism with a strength exponentially increasing 
with the number of particles. Our approach starts precisely with the 
properties of the stabilizer, in order to identify the useful correlations 
provided by quantum mechanics. This is in contrast to previous approaches 
that were either too general, e.g. Bell inequalities for general multiparticle states \cite{popescurohrlich,wwzb}, or too restricted, considering only few 
specific examples of hypergraph states and leading to non robust criteria \cite{Otfried}.
The violation of local realism is the key to further applications 
in information processing: Indeed, it is well known 
that violation of a Bell inequality leads to advantages, 
in distributed computation scenarios \cite{brunnerreview, brukner}. 
In addition, we will explicitly show that certain classes of hypergraph 
states lead to Heisenberg scaling in quantum metrology and advantages in measurement-based quantum computation.

%%%%%%%%%%%%%%%%%%%%%%%%%%%%%%%%%%%%

%%%%%%%%%%%%%%%%%%%%%%%%%%%%%%%%%%%%

{\it Hypergraph states.---} A hypergraph
$H=(V,E)$ consists of a set of vertices $V=\{1,...,N\}$ and a set 
of hyperedges $E\subset 2^V$, with $2^V$ the power set of $V$. While for graphs edges connect only two vertices, 
hyperedges can connect more than two vertices; 
examples of hypergraphs are depicted in Fig.~\ref{fig-hgbild}. For any hypergraph we define the corresponding hypergraph state $\ket{H}$ as the $N$-qubit 
state 
\begin{equation}
\ket{H}=\prod_{e\in E} C_e\ket{+}^{\otimes N},
\label{eq-hg-creation}
\end{equation}
where $\ket{+}=(\ket{0}+\ket{1})/\sqrt{2}$, $e$ is a hyperedge and $C_e$ is a multi-qubit phase gate acting on the Hilbert space associated with the vertices $v\in e$,
given by the matrix $C_e =\mathbbm {1}  - 2\ket{1\dots 1}\bra{1\dots 1}$. 
The first nontrivial hypergraph state consists of $N=3$ qubits 
connected by a single hyperedge [see Fig.~\ref{fig-hgbild}(a)]. Hypergraph states 
have been recognized as special cases of LME states, generated via a fixed interaction phase of $\phi=\pi$ \cite{Kruszynska2009}.

Alternatively, we can define the hypergraph states using a stabilizer 
formalism \cite{Otfried}.  For each qubit $i$
we define the operator
\begin{equation}
g_i=X_i\bigotimes_{e\in E} C_{e\backslash \{i\}}.
\label{eq-hg-stabilizer}
\end{equation}
Here and in what follows, we denote by $X_i$ and $Z_i$ the Pauli matrices, acting on $i^{th}$ qubit.  The hypergraph state can be defined as the unique eigenstate for all of 
them, $g_i \ket{H}=\ket{H}$ with the eigenvalue $+1$.  Consequently, the 
hypergraph state is an eigenstate of the entire stabilizer, i.e., the 
commutative group formed by all the products of the $g_i$. It should 
be noted that the  $g_i$ are, in general, non-local operators, 
as they are not tensor-products of operators acting on single parties. 
We say that a hyperedge has \textit{cardinality} $k$, 
if it circumscribes $k$ vertices and a hypergraph is \textit{$k$-uniform}, 
if all edges are $k$-edges. Finally, note that different 
hypergraphs may lead to equivalent hypergraph states, in the sense
that the two states can be converted into one other by a local basis 
change. For small numbers of qubits, the resulting equivalence classes
have been identified \cite{Otfried}.

%%%%%%%%%%%%%%%%%%%%%%%%%%%%%%%%%%%%

{\it Local correlations from the nonlocal stabilizer.---}
The key observation for the construction of our nonlocality arguments is
that the stabilizer of hypergraph states, despite being nonlocal, predicts
perfect correlations for some local measurements. In the following, we explain
this for the three-qubit hypergraph state $\ket{H_3}$, but the method is general.
The stabilizing operators for the three-qubit hypergraph state are
\begin{equation}\label{1}
g_1=X_1\otimes C_{23}, \quad g_2=X_2\otimes C_{13}, \quad g_3=X_3\otimes C_{12}.
\end{equation} 
We can expand the controlled phase gate $C_{ij}$ on two qubits, leading to 
\begin{equation}
\label{phasegate}
g_1=X_1\otimes (\ketbra{00}+\ketbra{01}+\ketbra{10}-\ketbra{11})
\end{equation}
and similar expressions for the other $g_i$. Since $g_1 \ket{H_3}=+\ket{H_3}$, 
the outcomes for $X$ measurements on the first qubit and $Z$ measurements 
on the second and third qubits are correlated: if one measures $"+"$ on 
the first qubit, then the other two parties cannot both measure $"-"$ in 
$Z$ direction, as this would produce $-1$ as the overall eigenvalue. So, we 
extract the first correlation from the stabilizer formalism:  
\begin{equation}
\label{A1}
P(+--|XZZ)=0.
\end{equation}
The l.h.s of Eq.~(\ref{A1}) denotes the probability of measuring 
$+--$ in $XZZ$ on the qubit $1,2$, and $3$, respectively.
Similarly, it follows that  if one measures $"-"$ in $X$ direction 
on the first qubit, then the other  parties, both have to measure 
$"-"$ in $Z$ direction. So we have:
\begin{equation}
\label{A}
P(-++|XZZ)+ P(-+-|XZZ)+P(--+|XZZ)=0,
\end{equation}
which implies, of course, that each of the probabilities is zero.
Since the three-qubit hypergraph state is symmetric, the same correlations 
for measuring $X$ on other qubits can be obtained by considering $g_2$ 
and $g_3$, leading to permutations of correlations  in Eq.~(\ref{A1} ,\ref{A}).

{\it The three-qubit hypergraph state $\ket{H_3}$.---}
We start with the discussion
of fully local hidden variable (HV) models. Such models assign 
for any value of the HV $\lambda$ results to all 
measurements of the parties in a local manner, meaning that the 
probabilities for a given HV factorize. If we denote by $r_i$ the 
result and by $s_i$ the measurement setting on the $i^{th}$ particle, 
respectively, then the probabilities coming from local models 
are of the form
\begin{align}
P & (r_1,r_2,r_3  |s_1,s_2,s_3)
=
\\
= &
\int d\lambda p(\lambda)
\chi^A(r_1|s_1,\lambda)
\chi^B(r_2|s_2,\lambda)
\chi^C(r_3|s_3,\lambda)
\nonumber.
\end{align}
For probabilities of this form, it is well known that it 
suffices to consider models which are, for a given $\lambda$,
deterministic. This means that  $\chi^i$ takes only the values
$0$ or $1$, and there is only a finite set of $\chi^i$ to 
consider.

\noindent
{\bf Observation 1.} 
{\it If a fully local hidden variable model satisfies the conditions from 
Eq.~(\ref{A1}, \ref{A})  and their symmetric correlations coming from the permutations, then it must fulfill}
\begin{equation}
 P(+--|XXX)+P(-+-|XXX)+P(--+|XXX) = 0.
\end{equation}
The proof of this statement is done by exhausting all possible local deterministic
assignments.
 
In contrast, for $\ket{H_3}$ we have   
\begin{align}
P(+--|XXX)=\frac{1}{16}
\label{eq-h3corr}
\end{align}
and the same holds for the permutations of the qubits. 
The above is a so-called Hardy argument \cite{Hardy92}, namely, a set of joint probabilities equal to $0$ or, equivalently, logical implications that together implies that some other probability is equal to zero. 

Our method shows how the correlations of the nonlocal stabilizer can be used 
for Hardy-type arguments. We recall that Hardy-type arguments have
been obtained for all permutation-symmetric states \cite{Wang, Abramsky}. 
However, they involved  different settings and have no direct connection 
with the stabilizer formalism, making a generalization complicated. In 
contrast, we will see that our measurements can even be used to prove 
genuine multiparticle nonlocality of the hypergraph state. 
First, we  translate the Hardy-type argument into a Bell inequality:

\noindent
{\bf Remark 2.}
{\it Putting together all the null terms derived from the  stabilizer formalism 
and subtracting the terms causing a Hardy-type argument, we obtain the Bell 
inequality
\begin{align}
\label{3correlations}
\langle \BB_3^{(1)} \rangle&= \big[  P(+--|XZZ)+  P(-++|XZZ)
\nonumber
\\
+&P(-+-|XZZ)+P(--+|XZZ)+ \mbox{ permutat.} \big]
\nonumber
\\
-&[ P(+ --|XXX)+ \mbox{ permutations }\big] \geq 0,
\end{align}
where the permutations include all distinct terms that are obtained
by permuting the qubits. The three-uniform hypergraph state violates 
the inequality (\ref{3correlations}) with the value of $\langle \BB_3^{(1)} \rangle=-3 / 16$.}

This Bell inequality follows from the Hardy argument:  If a deterministic
local model predicts one of the results with the minus signs, it 
also has to predict at least one of the results corresponding to the 
terms with a plus sign, otherwise it contradicts with the Hardy argument. In addition, all the terms with a minus sign 
are exclusive, so a deterministic LHV model can predict only 
one of them.

The Hardy-type argument and the Bell inequality can be generalized to a higher 
number of qubits, if we consider $N$-qubit hypergraphs with the single 
hyperedge having a cardinality $N$:

\noindent
{\bf Observation 3.}
{\it Consider the $N$-qubit hypergraph state with a single hyperedge
of cardinality $N$. Then, all the correlations coming from the 
stabilizer [as generalizations of Eqs.~(\ref{A1},\ref{A})] imply 
that for any possible set of results $\{r_i\}$ where one $r_{i_1}=+1$
and two $r_{i_2}= r_{i_3} = -1$ one has
\begin{equation}
 P(r_1, r_2, ..., r_N|X_1 X_2 ... X_N) = 0.
\end{equation}
For the hypergraph state, however, this probability equals 
$1/2^{(2N-2)}.$ This Hardy-type argument leads to a Bell 
inequality as in Eq.~(\ref{3correlations}) which is violated
by the state with a value of $-(2^N-N-2)/2^{(2N-2)}.$}

Clearly, the violation of the Bell inequality is not strong, as it does
not increase with the number of particles. Nevertheless, Observation~3 
shows that the nonlocal stabilizer formalism allows one to easily obtain nonlocality
proofs. In fact, one can directly derive similar arguments for other hypergraph states (e.g. states with one hyperedge of cardinality $N$ and one 
further arbitrary hyperedge), these results will be presented elsewhere. Note
that these states are not symmetric, so the results of Refs.~\cite{Wang, Abramsky}
do not apply. 

So far, we only considered fully local models, where for a 
given HV all the probabilities factorise. Now we go 
beyond this restricted type of models to the so-called hybrid models \cite{Svetlichny87}. 
We consider a  bipartition of the three particles, say $A|BC$, and 
consider a model of the  type
$
P  (r_1,r_2,r_3 |s_1,s_2,s_3)
=
\int d\lambda p(\lambda)
\chi^A(r_1|s_1,\lambda)
\chi^{BC}(r_2, r_3|s_2, s_3,\lambda).
$
Here, Alice is separated from the rest, but $\chi^{BC}$ may contain 
correlations, e.g., coming from an entangled state between $B$ and $C$.
In order to be physically reasonable, however, we still request 
$\chi^{BC}$ not to allow instantaneous signaling.

This kind of models, even if different bipartitions are mixed, cannot explain 
the correlations of the hypergraph state, meaning that the hypergraph state 
is genuine multiparticle nonlocal. First, one can see by direct inspection 
that the stabilizer conditions from Eqs.~(\ref{A1}, \ref{A})
are not compatible with the hypergraph correlations $P(---|XXX) = 1/16$ 
and  $P(---|ZZZ) = 1/8$. Contrary to the correlations in Eq.~(\ref{eq-h3corr})
these are symmetric, and allow the construction of a Bell-Svetlichny 
inequality \cite{Svetlichny87} valid for all the different bipartitions:

\noindent
{\bf Observation 4.}
{\it 
Putting all the terms from the hypergraph stabilizer formalism 
and the correlations $P(---|XXX)$ and $P(---|ZZZ)$ together, we obtain 
the following Bell-Svetlichny inequality for  genuine multiparticle  nonlocality,
\begin{align}
\label{3gen}
\langle\BB_3^{(2)}\rangle&=
 \big[P(+--|XZZ)+  P(-++|XZZ)
\nonumber
\\
+&P(-+-|XZZ)+P(--+|XZZ)+ \mbox{ permutat.} \big]
\nonumber
\\
+&P(- --|XXX) -   P(- - -|ZZZ) \geq 0,
\end{align}
which is violated by the state $\ket{H_3}$ with $\langle\BB_3^{(2)}\rangle=-1 \backslash 16$.}

The proof is done by an exhaustive assignments of nonsignaling 
and local models.

To investigate the noise tolerance of Ineq.~\eqref{3gen}, we 
consider states of the type
$
\varrho= (1-\varepsilon)\ket{H}\bra{H}+\varepsilon \mathbbm{1}/8
$
and ask how much noise can be added, while the inequality is still 
violated. The white noise tolerance of the Ineq.~(\ref{3gen})  
is  $\varepsilon = {1}/{13} \approx 7.69\%$  and  is optimal 
in the sense that for larger values of $\varepsilon$ a hybrid model 
can be found which explains all possible measurements of $X$ and $Z$ 
(within numerical precision). The existence of such a model can be 
shown by linear programming (see Appendix A \cite{appremark}). With the same 
method we can also prove that the state becomes fully local with respect 
to $X$ and $Z$ measurements for $\varepsilon \geq  2/3  \approx 66.6\% $.

%%%%%%%%%%%%%%%%%%%%%%%%%%%%%%%%%%%%

{\it Three uniform hypergraph states.---} 
Let us extend our analysis to hypergraph states with a 
larger number of particles. Here, it is interesting to ask 
whether the violation of Bell inequalities increases exponentially 
with the number of parties. Such a behaviour has previously been 
observed only for GHZ states \cite{Mermin} and some cluster states 
\cite{gthb}. 

GHZ states are described by fully connected graphs 
(see Fig.~\ref{fig-hgbild}), i.e., fully connected 
two-uniform hypergraph states. It is thus natural to start 
with fully connected three-uniform hypergraph states. First, 
we observe that for such states on $N$ qubits and 
for even $m$ with $1<m<N$
\begin{equation}\label{lemma3}
\langle \underset{m}{\underbrace{X\dots X}}Z\dots Z\rangle =\begin{cases}
\begin{array}[t]{cc}
+\frac{1}{2} & \mbox{\ensuremath{\mbox{if \ensuremath{m=2} mod \ensuremath{4}}}},\\
-\frac{1}{2} & \mbox{\ensuremath{\mbox{if \ensuremath{m=0} mod \ensuremath{4}}}}.
\end{array}\end{cases}
\end{equation}
Moreover, if $m=N$, then the correlations are given by 
\begin{equation}\label{lemma33}
\langle \underset{N}{\underbrace{XX\dots XX}}\rangle =\begin{cases}
\begin{array}[t]{cc}
0 & \mbox{\ensuremath{\mbox{if \ensuremath{N=0} mod \ensuremath{4}}}},\\
1  & \mbox{\ensuremath{\mbox{if \ensuremath{N=2} mod \ensuremath{4}}}}.
\end{array}\end{cases}
\end{equation}
Finally, we always have 
$\langle{{ZZ\dots ZZ}}\rangle =0$ (see Appendix B for details \cite{appremark}).
 
We then consider the following Bell operator 
\begin{align}
\label{bell2}
\BB_N & =  - \big[AAA\dots AA\big] +  
\big[BBA\dots A   +  \;\mbox{permutat.}\big] 
- \nonumber\\
  - & \big[BBBBA\dots A  + \;\mbox{permutat.}\big] +
 \big[\dots\big]  - \dots 
\end{align}
Note that this Bell operator is similar to  $\langle \BB_N^M \rangle $ 
of the original Mermin inequality \cite{Mermin}, but it differs in the 
number of $B$ (always even) considered. Using the correlations computed 
above, we can state:

\noindent
{\bf Observation 5.}
{\it 
If we fix in the Bell operator $\BB_N$ in Eq.~(\ref{bell2}) the measurements 
to be $A=Z$ and $B=X$, then the $N$-qubit fully-connected three-uniform 
hypergraph state violates the classical bound, by an amount that grows 
exponentially with number of qubits, namely
\begin{align}
\left\langle \BB_N \right\rangle_C &\leq 2^{\left\lfloor N/2\right\rfloor} \quad  \mbox{ for local  HV models, and}
\nonumber
\\
\left\langle \BB_N \right\rangle_Q &\geq 2^{N-2}-\frac{1}{2} \quad  \mbox{ for the hypergraph state.}
\end{align}
}
The proof is given in Appendix C \cite{appremark}.

%%%%%%%%%%%%%%%%%%%%%%%%%%%%%%%%%%%%

%%%%%%%%%%%%%%%%%%%%%%%%%%%%%%%%%%%%

{\it Four-uniform hypergraph states.---}
Finally, let us consider four-uniform complete hypergraph states. 
For them, the correlations of measurements
as in Eq.~(\ref{lemma3}) are not so simple: They are not constant, and depend
on $m$ as well as on $N$. Nevertheless, they can be explicitly computed, and
detailed formulas are given in the Appendix D \cite{appremark}. 
{From} these correlations, we can state:

\noindent
{\bf Observation 6.}
{\it The $N$-qubit fully-connected four-uniform hypergraph state violates
 local realism by an amount that grows exponentially with number
of qubits. More precisely, one can find a Mermin-like Bell operator
$\BB_N$ such that
\begin{align}
\frac{\left\langle \BB_N \right\rangle_Q}
{\left\langle \BB_N \right\rangle_C}
\stackrel{N\rightarrow \infty}{\sim}   
\frac{\Big(1+\frac{1}{\sqrt{2}}\Big)^{N-1} }{\sqrt{2}^{N+3}}
\approx 
\frac{1.20711^{N}}{2\sqrt{2}+2}.
\end{align}}
A detailed discussion is provided in Appendix E \cite{appremark}.

%Similar results can be obtained for five-uniform hypergraph states. 

%%%%%%%%%%%%%%%%%%%%%%%%%%%%%%%%%%%%
{\it Robustness.---}
So far, we have shown that three- and four-uniform hypergraph states
 violate  local realism comparable to GHZ states. 
A striking difference is, however, that the entanglement
and Bell inequality violation of hypergraph states is robust under 
particle loss. This is in stark contrast to GHZ states, which become fully
separable if a particle is lost. We can state:

\noindent
{\bf Observation 7.}
{\it The $N$-qubit fully-connected four-uniform hypergraph 
state preserves the violation of the local realism even after 
loss of one particle. More precisely, for $N=8k+4$ we have
${\left\langle \BB_{N-1}\right\rangle_Q}/
{\left\langle \BB_N \right\rangle_Q}
\stackrel{N\rightarrow \infty}{\sim} {1}/({\sqrt{2}+1}).$
This means that the reduced state shows the same exponential scaling
of the Bell inequality violation as the original state.
}

For the detailed discussions see Appendix F \cite{appremark}.

For three-uniform complete hypergraph states we can prove that
the reduced states are highly entangled, as they violate inequalities
testing for separability \cite{Roy} exponentially. This violation 
decreases with the number of traced out qubits, but persists
even if several qubits are lost. This suggests that 
this class of hypergraph states is also more robust than 
GHZ states, details can be found in Appendix G \cite{appremark}. Despite the structural differences, this property resembles  of the W state, which is itself less entangled but more robust than the GHZ state \citep{Buzek}. In addition, this may allow  the lower detection efficiency in the experiments.

{\it Discussion and conclusion.---}
A first application of our results is quantum 
metrology. In the standard scheme of quantum  metrology one 
measures an observable $M_\theta$ and tries to determine the 
parameter $\theta$ which describes a phase in some 
basis \cite{giovanetti, weibo}. If one takes product states, 
one obtains a signal $\mean{M_\theta} \sim \cos{(\theta)}$ on 
a single particle, repeating it on $N$ particles allows to determine 
$\theta$ with an accuracy $\delta \theta \sim 1/\sqrt{N}$, the so-called 
standard quantum limit. Using an $N$-qubit GHZ state, however, one observes $\mean{(M_\theta)^{\otimes N}} \sim \cos{(N\theta)}$ and this phase 
super-resolution  allows to reach the Heisenberg limit 
$\delta \theta \sim 1/{N}$. For a general state $\vr$, 
it has been shown that the visibility of the phase super-resolution
is given by the expectation value of the Mermin-type inequality, 
$V=Tr(\BB_N \vr)/2^{N-1}$ \cite{weibo}. So, since the three-uniform
hypergraph states violate this inequality with a value $\mean{\BB_N}_Q 
\sim 2^{(N-2)}$ the visibility is $V \sim 1/2$, independently of the number
of particles. This means that these states can be used for Heisenberg-limited
metrology, and from our results they can be expected to have the advantage 
of being more robust to noise and particle losses.

The second  application of exponential violation of Bell
inequalities is \textit{nonadaptive measurement based quantum
computation with linear side-processing}  ($NMQC_{\oplus}$) \cite{Hoban11}. 
$NMQC_{\oplus}$ is a non-universal model of quantum computation where linear classical
side-processing is combined with quantum measurements in a nonadaptive way, i.e., the choice of settings is independent of previous outcomes.
In Ref.~\cite{Hoban11} the authors connect the expectation value of a full-correlation Bell expression \cite{wwzb} with the success probability of computing a Boolean function, specified as a function of the inequality coefficients, via $NMQC_{\oplus}$. In particular, the exponential violation of generalized Svetlichny inequalities \cite{Collins02} (equal to Mermin inequalities for even $N$), corresponds to a constant success probability $P_{\rm succ}$ of computing the pairwise $\mathrm{AND}$ on $N$ bits extracted from a uniform distribution, whereas in the classical case $P_{\rm succ}-1/2$ decrease exponentially with $N$.
As a consequence, the exponential violation of the full-correlation Bell expression $\mathcal{B}_N$ can be directly related to an exponential advantage for computation tasks in the $NMQC_{\oplus}$ framework. Moreover, in several cases, e.g., $4$-uniform hypergraph states of $N = 6\ {\rm mod}\ 8$ qubits, also the Svetlichny inequality is violated exponentially, providing an advantage for computation of the pairwise $\mathrm{AND}$ discussed in Ref.~\cite{Hoban11}.   

In summary, we have shown that hypergraph states violate local realism in many ways. 
This suggests that they are interesting resources for quantum information
processing, moreover, this makes the observation of
hypergraph states a promising task for experimentalists. In our work, 
we focused only on some classes of hypergraph states, but for future 
research, it would be desirable to identify classes of hypergraph states 
which allow for an all-versus-nothing violation of local realism or which
are strongly genuine multiparticle nonlocal.

%%%%%%%%%%%%%%%%%%%%%%%%%%%%%%%%%%%%

We thank D.~Nagaj, F.~Steinhoff, M.~Wie\'sniak, and B.~Yoshida for 
discussions. This work has been supported by the EU 
(Marie Curie CIG 293993/ENFOQI), 
the FQXi Fund (Silicon Valley Community Foundation),  the DFG and the ERC.

%%%%%%%%%%%%%%%%%%%%%%%%%%%%%%%%%%%%%%%%%%%%%%%%%%%
\onecolumngrid

\section{Appendix A. Genuine multiparticle nonlocality with linear programming}\label{App:LP}

Fixed the number of measurement settings and outcomes, probabilities arising from a hybrid local-nonsignalling model, as the one described in the main text for the splitting $A|BC$, form a polytope whose extremal points are given by combination of deterministic local assignments for the party $A$ and extremal nonsignalling assignments, i.e., local deterministic and PR-boxes \cite{brunnerreview}, for the parties $BC$. In order to detect genuine multiparticle nonlocality, one has to consider all combinations of probabilities arising from the other local-nonsignalling splitting, namely $C|AB$ and $B|AC$. Geometrically, this corresponds to take the convex hull of the three polytopes associated with the three different splitting. Let us denote such polytopes as $\mathcal{P}_{A|BC}, \mathcal{P}_{C|AB} ,\mathcal{P}_{B|AC}$ and their convex hull as $\mathcal{P}_{L-NS}$. By definition of convex hull, every vector $\mathbf{p}\in \mathcal{P}_{L-NS}$ can be written as a convex combination of three vectors $\mathbf{p}_{
A|BC}, 
\mathbf{p}_{C|AB}$ and $\mathbf{p}_{B|AC}$, which, in turn, can be written as a convex combination of the vertices of the corresponding polytope.

To check whether a given point $\mathbf{p}$ belongs to $\mathcal{P}_{L-NS}$ it is, therefore, sufficient the description in terms of the extremal points of  $\mathcal{P}_{A|BC}, \mathcal{P}_{C|AB} ,\mathcal{P}_{B|AC}$. Let us denote them as $\{\mathbf{v}_i\}$. The membership problem can then be formulated as a linear program (LP) \cite{brunnerreview}

\begin{equation}\begin{split}\label{e:lp}
\text{maximize: } &\mathbf{\lambda} \cdot \mathbf{p} - C \\
\text{subject to: }&
 \mathbf{\lambda} \cdot \mathbf{v}_i - C  \leq 0\text{ , for all }\mathbf{v}_i \\
 & \mathbf{\lambda} \cdot \mathbf{p} - C \leq 1. 
\end{split}\end{equation}

The variable of the LP are $\{\mathbf{\lambda},C\}$, where $\mathbf{\lambda}$ represents the coefficient of a Bell-Svetlichny inequality, detecting genuine multiparticle nonlocality, and $C$ the corresponding local-nonsignalling bound. The LP optimizes the coefficients $\mathbf{\lambda}$ to obtain the maximal value (at most $C+1$) for the quantum probabilities, while keeping the local-nonsignaling bound $C$. As a consequence, the vector $\mathbf{p}$ can be written as a convex combination of $\{\mathbf{v}_i\}$ if and only if the optimal value of the LP is 0.

The noise tolerance for $\ket{H_3}$ can then be computed by mixing it with white noise, i.e., $\ket{H_3}\bra{H_3}\mapsto
\varrho= (1-\varepsilon)\ket{H}\bra{H}+\varepsilon \mathbbm{1}/8$, and compute for which values of $\varepsilon$ the LP \eqref{e:lp} gives optimal value $0$. Standard numerical techniques for LP give that up to $\varepsilon\approx {1}/{13} \approx 7.69\%$ the probabilities for $X$ and $Z$ measurements cannot be explained by a hybrid local-nonsignalling model.

\section{APPENDIX B: Correlations for three-uniform hypergraph states}

\subsection{B.1. Preliminary calculations}

Before starting with the actual calculations, we need to settle couple 
of identities and a look-up table,  which we will refer to throughout 
the main proofs.

The first and probably the most important identity is a commutation relation
between multi-qubit phase gates and Pauli X matrices \cite{Otfried},
\begin{equation}{\label{identity}}
C_e\big(\bigotimes_{i \in K} X_k \big)=(\bigotimes_{i \in K} X_k \big)\big(\prod_{f\in \mathcal{P}(K)}C_{e\backslash\{f\}}\big).
\end{equation}
Here, $\mathcal{P}(K)$ denotes the power set of the index set $K$.
Note that the product of the $C_{e\backslash\{f\}}$ may include
the term $C_{\emptyset}$, which is defined to be $-\openone$
and leads to a global sign.

Furthermore, it turns out to be useful to recall some basic facts
about binomial coefficients, as the appear frequently in the following
calculations.

\begin{lemma}
The following equalities hold:
\begin{equation}\label{identity1}
Re\Big[(1+i)^n\Big]=\sum_{k=0,4,\dots}^n\binom{n}{k}-\binom{n}{k+2},
\end{equation}
\begin{equation}\label{identity2}
Im\Big[(1+i)^n\Big]=\sum_{k=0,4,\dots}^n\binom{n}{k+1}-\binom{n}{k+3}.
\end{equation}
\end{lemma}
\begin{proof}
Here we  derive (\ref{identity1}) and  (\ref{identity2})  together:
\begin{equation}
s:=(1+i)^n=\sum^n_{k=0} \binom{n}{k}i^k =\sum^n_{k=0,4,\dots} \binom{n}{k}+ i\binom{n}{k+1} -\binom{n}{k+2}-i\binom{n}{k+3}.
\end{equation}
It is easy to spot that $Re[s]$ and $Im[s]$ indeed leads to the identities (\ref{identity1}) and  (\ref{identity2}) respectively.
\end{proof}

The following look-up table represents the values of (\ref{identity1}) and  (\ref{identity2})  for different $n$. These values can be derived from the basic properties of complex numbers:

\begin{center}
\begin{tabular}{|c|c|c|c|c|c|}
\hline 
\# & $n$ & $Re\Big[(1+i)^{n}\Big]$ & $Im\Big[(1+i)^{n}\Big]$ & $Re\Big[(1+i)^{n}\Big]$$+Im\Big[(1+i)^{n}\Big]$ & $Re\Big[(1+i)^{n}\Big]$$-Im\Big[(1+i)^{n}\Big]$\tabularnewline
\hline 
\hline 
1. & $n=0$ mod 8 & $+2^{\frac{n}{2}}$ & $0$ & $+2^{\frac{n}{2}}$ & $+2^{\frac{n}{2}}$\tabularnewline
\hline 
2. & $n=1$ mod 8 & $+2^{\frac{n-1}{2}}$ & $+2^{\frac{n-1}{2}}$ & $+2^{\frac{n+1}{2}}$ & $0$\tabularnewline
\hline 
3. & $n=2$ mod 8 & $0$ & $+2^{\frac{n}{2}}$ & $+2^{\frac{n}{2}}$ & $-2^{\frac{n}{2}}$\tabularnewline
\hline 
4. & $n=3$ mod 8 & $-2^{\frac{n-1}{2}}$ & $+2^{\frac{n-1}{2}}$ & $0$ & $-2^{\frac{n+1}{2}}$\tabularnewline
\hline 
5. & $n=4$ mod 8 & $-2^{\frac{n}{2}}$ & $0$ & $-2^{\frac{n}{2}}$ & $-2^{\frac{n}{2}}$\tabularnewline
\hline 
6. & $n=5$ mod 8 & $-2^{\frac{n-1}{2}}$ & $-2^{\frac{n-1}{2}}$ & $-2^{\frac{n+1}{2}}$ & $0$\tabularnewline
\hline 
7. & $n=6$ mod 8 & $0$ & $-2^{\frac{n}{2}}$ & $-2^{\frac{n}{2}}$ & $+2^{\frac{n}{2}}$\tabularnewline
\hline 
8. & $n=7$ mod 8 & $+2^{\frac{n-1}{2}}$ & $-2^{\frac{n-1}{2}}$ & $0$ & $+2^{\frac{n+1}{2}}$\tabularnewline
\hline 
\end{tabular}\\
$\quad$\\

\textbf{Table 0.} Look-up table for the values of (\ref{identity1}) and  (\ref{identity2}) .
\end{center}

\subsection{B.2. Correlations for X and Z measurements
on fully-connected three-uniform hypergraph states}

\begin{lemma} \label{lemma3full} Consider an arbitrary number of qubits $N$ 
and the three-uniform fully-connected hypergraph (HG) states. Then, if
$m$ is even with $1<m<N$ the following equality holds 
\begin{equation}
\langle \underset{m}{\underbrace{X\dots X}}Z\dots Z\rangle =\begin{cases}
\begin{array}[t]{cc}
+\frac{1}{2} & \mbox{\ensuremath{\mbox{if \ensuremath{m=2} mod \ensuremath{4}}}},\\
-\frac{1}{2} & \mbox{\ensuremath{\mbox{if \ensuremath{m=0} mod \ensuremath{4}}}}.
\end{array}\end{cases}
\end{equation}
\end{lemma}
\begin{proof}  We can write:

\begin{equation}\label{appendixAeq1}
K:=\bra{HG} X\dots XZ\dots Z\ket{HG}=\bra{+}^{\otimes N}\bigg(\prod_{e\in E} C_e\bigg) X\dots XZ\dots Z\bigg( \prod_{e\in E} C_e\bigg)  \ket{+}^{\otimes N}.
\end{equation}
We  can group all the controlled phase gates on the right hand side of the expression (\ref{appendixAeq1}). Note that the operators $C_e$ and $X_i$ do not commute, but we can use the identity  ($\ref{identity}$). While regrouping we count the multiplicity of  each phase gate. If each phase gate appears  even times, we get an identity as $C^2=\mathbbm{1}$, if not, we keep these phase gates with the multiplicity one for the further calculations.

For the purposes which will become apparent shortly, we denote the parties which measure in $X$ direction by $\circledast$ and ones in $Z$ direction by $\bigtriangleup$, in a way that, for example, if an arbitrary phase gate acts on $XXXZZ$, it is represented as $\circledast\circledast\circledast\bigtriangleup\bigtriangleup$. 
 Without  loss of generality, we fix one phase gate $C_e$ and consider all the possible scenarios of $\circledast$ and $\bigtriangleup$ it can be acting on. Since we work on three-uniform HG states, every phase gate acts on bashes of different three party systems. These parties can be either of type $\circledast$ or $\bigtriangleup$ and we have to consider all possible scenarios.  Since we are working with  symmetric states, we can sort the parties such that we have  $m$ $\circledast$'s  followed by $(N-m)$ $\bigtriangleup$'s:
 \begin{equation} \label{appendixAeq2}
 \underset{m}{\underbrace{\circledast\dots \circledast}} \underset{N-m}{\underbrace{\bigtriangleup\dots\bigtriangleup}} 
\end{equation}  
 and then  we can enumerate all the scenarios of one phase gate acting on  (\ref{appendixAeq2}): \\
 
 \begin{enumerate}
 \item $C_eZZZ $  corresponds to $\bigtriangleup\bigtriangleup\bigtriangleup \quad\quad\quad\quad$ $3. \; C_eXXZ$ corresponds to $\circledast\circledast\bigtriangleup$\\
 \item $C_eXZZ$ corresponds to  $\circledast\bigtriangleup\bigtriangleup\quad\quad\quad\quad$  $4.\; C_eXXX$ corresponds to $\circledast\circledast\circledast$\\
\end{enumerate}  

We consider each case separately:\\

1.  $C_eZZZ=ZZZC_e$ as $C_e$ and $Z$ commute.  $C_e$ moves on the right side with the multiplicity one. To save us writing in the future, we will denote the multiplicity of the phase gate moving on the right side by $\# e$. In this particular case it is $ \#  \bigtriangleup\bigtriangleup\bigtriangleup =1$.  However, on the right side of the equation (\ref{appendixAeq1})  we have a product of all three-party phase gates. Therefore, we get $C_e$ with the multiplicity of two and $C_e^2=\mathbbm{1}$. Note that, as we have chosen an arbitrary three-qubit phase gate, the same result holds for every such phase gate. So, all three-qubit phase gates coming from the case 1,  cancel out. We will see that with the same reasoning all three qubit phase gates cancel out (give an identity). \\

2. For $C_eXZZ$, we use the   identity   ($\ref{identity}$):
\begin{equation} \label{appendixAeq3}
C_eXZZ=XC_eC_{\{e\backslash X\}}ZZ.
\end{equation}
The three-qubit phase gate $C_e$,  from (\ref{appendixAeq3}),   appears with the multiplicity one $(\# \circledast\bigtriangleup\bigtriangleup=1)$ and like in the case 1, it gives an identity when being multiplied by the same three-qubit phase gate at the right side of the expression in (\ref{appendixAeq1}).  It is more tricky to calculate the multiplicity of $C_{\{e\backslash X\}}$ ($\# \bigtriangleup\bigtriangleup$ as the $\circledast$ part (or equivalently, $X$ part) is removed from the set of vertices $e$.). For this we need to fix $\bigtriangleup\bigtriangleup$ and count all the scenarios when an arbitrary $C_{e}$ is reduced to $\bigtriangleup\bigtriangleup$. As we are working with the symmetric case, such scenario repeats $\binom{m}{1}=m$ times, where $m$ is the number of parties measuring in $X$ direction.We shortly denote this as $\# \bigtriangleup\bigtriangleup =\binom{m}{1}=m$. So, as $m$ is an even number,   $(C_{\{e\backslash X\}})^m=(C_{\bigtriangleup\bigtriangleup})^m=\mathbbm{1}$.\\
Note that the gate $C_{\bigtriangleup\bigtriangleup}$ can only be generated from in case 2.\\

3. For $C_eXXZ$, we use the   identity   ($\ref{identity}$):
\begin{equation} \label{appendixAeq4}
C_eXXZ=XXC_eC_{\{e\backslash XX\}} \Big[\prod_{\forall X} C_{\{e\backslash X\}}\Big] Z.
\end{equation}
The three-qubit phase gate $C_e$,  from (\ref{appendixAeq3}),   appears with the multiplicity one $(\# \circledast\circledast\bigtriangleup=1)$; therefore,  like  in the two previous cases,  it cancels out on the right side of the expression. A multiplicity of $C_{\{e\backslash XX\}}$ is calculated by fixing a concrete $\bigtriangleup$ and counting all possible appearance of arbitrary $\circledast\circledast$. As the number of parties measuring in direction is $X$ is $m$, this means that it is all combination of two parties with $X$ measurements out of total $m$ parties. So, 
\begin{equation}
 \# \bigtriangleup=\binom{m}{2}=\frac{m(m-1)}{2}=\begin{cases}
\begin{array}[t]{ccc}
\mbox{even, if} & m=0 \; mod\; 4 & \Rightarrow C_{\triangle}\mbox{ cancels out.}\\
\mbox{odd, if } & m=2 \; mod \; 4 & \Rightarrow C_{\triangle}\mbox{ remains.}
\end{array}\end{cases}
\end{equation}
The last one from this case is the multiplicity of $C_{\{e\backslash X\}}$ or $\# \circledast\bigtriangleup$. Here we fix one qubit from $m$ ($X$ direction) and one from $N-m$  ($Z$ direction) and count the number of such occurrences, when the third qubit is an arbitrary one of the type  $\circledast$, which is exactly $\binom{m-1}{1}$. Therefore, 
\begin{equation}
\#\circledast\bigtriangleup=\binom{m-1}{1}=m-1,\mbox{which is odd},\Rightarrow\; C_{\circledast\bigtriangleup}\mbox{remains}.
\end{equation}\\

4. For $C_eXXX$, we use the   identity   ($\ref{identity}$):
\begin{equation} \label{appendixAeq5}
C_eXXX=XXXC_e  \Big[\prod_{\forall X} C_{\{e\backslash X\}}\Big]  \Big[\prod_{\forall X} C_{\{e\backslash XX\}}\Big]C_{\{\}}.
\end{equation}
$C_e$ occurs once and it gives an identity with the other one from the right side like in the previous cases. The multiplicity of $C_{\{e\backslash X\}}$ is $\# \circledast\circledast$. Here we fix two parties in $X$ direction and count the occurrence of this scenario by altering the third party, from the remaining $m-2$, in $X$ direction. Therefore,
\begin{equation}
\#\circledast\circledast=\binom{m-2}{1}=m-2,\mbox{which is even},\Rightarrow\; C_{\circledast\circledast}\mbox{ cancels out}.
\end{equation}
Similarly for $C_{\{e\backslash XX\}}$,  we fix one party in $X$ direction and count all possibilities of choosing two parties out of remaining $m-1$. Therefore,
\begin{equation}
 \# \circledast=\binom{m-1}{2}=\frac{(m-1)(m-2)}{2}=\begin{cases}
\begin{array}[t]{ccc}
\mbox{odd, if} & m=0 \; mod\; 4 & \Rightarrow\;  C_\circledast \mbox{ remains.}\\
\mbox{even if, } & m=2\; mod\; 4 & \Rightarrow\;   C_\circledast \mbox{ cancels out.}
\end{array}\end{cases}
\end{equation}
At last, we consider $C_{\{\}}$. This gate determines the global sign of the expectation value and it appears only  when $C_e$ acts on systems which are all measured in $X$ direction.  Therefore, 
\begin{equation}{\label{3unifGL}}
\#\{\}=\binom{m}{3}=\frac{m(m-1)(m-2)}{2 \cdot 3},\mbox{which is even}\Rightarrow\; \mbox{ the  global sign is positive }.
\end{equation}\\

To go on, we need to consider two cases \#1: $m=0$ mod 4 and  \#2: $m=2$ mod 4  and calculate the expectation value separately for both: \\

\textbf{Case \#1:}  When $m=0 \; mod\;  4$, we write out all remaining phase gates and continue the derivation from  equation (\ref{appendixAeq1}):
\begin{equation}
\langle K \rangle :=\bra{+}^{\otimes N} X\dots X  Z\dots Z  \prod_{\forall \circledast, \forall\bigtriangleup } C_{\circledast\bigtriangleup} C_{\circledast}\ket{+}^{\otimes N}.
\end{equation}\\

Using the fact that $X$ is an eigenstate of  $\bra{+}$, we can get rid of all $X$s and then we can write $C_{\bigtriangleup}$ instead of $Z$:
 \begin{align}\label{appendixAeq6}
 \begin{split}
\langle K \rangle= & \bra{+}^{\otimes N}\prod_{\forall \circledast, \forall\bigtriangleup } C_{\circledast\bigtriangleup} C_{\circledast}C_{\bigtriangleup}\ket{+}^{\otimes N}=\frac{1}{\sqrt{2}^N}\sum_{i=00\dots 00}^{11\dots 11}\bra{i}\prod_{\forall \circledast, \forall\bigtriangleup } C_{\circledast\bigtriangleup} C_{\circledast}C_{\bigtriangleup}  \frac{1}{\sqrt{2}^N}\sum_{j=00\dots 00}^{11\dots 11}\ket{j}\\
= & \frac{1}{2^N}\bigg[\bra{00\dots 00}\prod_{\forall \circledast, \forall\bigtriangleup } C_{\circledast\bigtriangleup} C_{\circledast}C_{\bigtriangleup}\ket{00\dots 00}+\dots +\bra{11\dots 11}\prod_{\forall \circledast, \forall\bigtriangleup } C_{\circledast\bigtriangleup} C_{\circledast}C_{\bigtriangleup}\ket{11\dots 11}\bigg]\\
= &\frac{1}{2^N}Tr\bigg[\prod_{\forall \circledast, \forall\bigtriangleup } C_{\circledast\bigtriangleup} C_{\circledast}C_{\bigtriangleup}\bigg].
\end{split}
\end{align}
In (\ref{appendixAeq6}), to get line two from the line one, note that $\prod_{\forall \circledast, \forall\bigtriangleup } C_{\circledast\bigtriangleup} C_{\circledast}C_{\bigtriangleup}$ is a diagonal matrix. \\

 To evaluate the trace of the given diagonal matrix, we need to find the difference between the  number of $+1$ and  $-1$ on the diagonal.  We write every row in the computational basis by enumerating it with the binary notation. For each  row, we denote by $\alpha$ the number of 1's in binary notation appearing in the first $m$ columns and by $\beta$, the same on the rest. For example, for $N=7$ , and $m=4$, the basis element $\ket{1101110}$ leads to  $\alpha$=3 and $\beta=2$.  Considering the phase gates in the equation (\ref{appendixAeq6}), the expression $(-1)^s$ defines whether in the given row the diagonal element is $+1$ or  $-1$, where : 
\begin{equation}
s:=\binom{\alpha}{1}\binom{\beta}{1}+\binom{\alpha}{1}+\binom{\beta}{1}=\alpha \beta +\alpha+\beta	.
\end{equation}
In $s$, $\binom{\alpha}{1}\binom{\beta}{1}$  denotes how many $C_{\circledast\bigtriangleup} $ acts on the row. Also, $\binom{\alpha}{1}$ determines the number of $C_{\circledast}$ and  $\binom{\beta}{1}$, number of $C_{\bigtriangleup} $. Every time when the phase gate acts, it changes the sign of the diagonal element on the row. Therefore, we need to determine the number s:\\

To see whether $s$ is even or odd, we have to  consider the following cases exhaustively:\\
$\quad$\\
\begin{tabular}{lllll}
\textbf{1. } & $\alpha$ is even \& $\beta$ is even & $(-1)^{s}=+1$ &  &  \multirow{2}{*}{$\Rightarrow$ These two cases sum up to zero.}\tabularnewline
\textbf{2. } & $\alpha$ is even \& $\beta$ is odd  & $(-1)^{s}=-1$ &  & \tabularnewline
\end{tabular}\\
$\quad$\\
\begin{tabular}{lllll}
\textbf{3. } & $\alpha$ is odd \& $\beta$ is even & $(-1)^{s}=-1$ &  & \multirow{2}{*}{$\Rightarrow$ These two cases contribute with the negative sign.}\tabularnewline
\textbf{4.} & $\alpha$ is odd \& $\beta$ is odd & $(-1)^{s}=-1$ &  & \tabularnewline
\end{tabular}\\

From the cases 3 and 4, one can directly calculate the trace:
 \begin{equation}\label{appendixAeq7}
\langle K \rangle=\frac{1}{2^N}\bigg[-\sum_{\alpha=1,3,\dots }^m\binom{m}{\alpha}\sum_{\beta=0}^{N-m}\binom{N-m}{\beta}\bigg]=-\frac{2^{m-1}2^{N-m}}{2^N}=-\frac{1}{2}.
\end{equation}
So, we get that if $m$ is divisible by 4, 
\begin{equation}
\langle \underset{m}{\underbrace{X\dots X}}Z\dots Z\rangle =-\frac{1}{2}.
\end{equation}

\textbf{Case \#2:}  We use the identical approach: when $m=2 \; mod \; 4 $,  we write out all remaining phase gates and continue the derivation from  equation (\ref{appendixAeq1}):
\begin{equation}\label{appendixAeq8}
\langle K \rangle=\bra{+}^{\otimes N} X\dots X  Z\dots Z  \prod_{\forall \circledast, \forall\bigtriangleup } C_{\circledast\bigtriangleup} C_{\bigtriangleup}\ket{+}^{\otimes N}.
\end{equation}\\
Again we use the fact that $X$ is an eigenstate of $\bra{+}$ and $Z=C_{\bigtriangleup}$. As in (\ref{appendixAeq8}), there is already one  $(C_{\bigtriangleup})$, they cancel. Therefore, we are left with: 
\begin{equation}\label{appendixAeq8}
\langle K \rangle=\bra{+}^{\otimes N} \prod_{\forall \circledast, \forall\bigtriangleup } C_{\circledast\bigtriangleup} \ket{+}^{\otimes N}=\frac{1}{2^N}Tr\bigg[\prod_{\forall \circledast, \forall\bigtriangleup } C_{\circledast\bigtriangleup} \bigg].
\end{equation}\\
We need to define the sign of the diagonal element by $(-1)^s$, where 
\begin{equation}
s=\binom{\alpha}{1}\binom{\beta}{1}=\alpha\beta.
\end{equation}

\begin{tabular}{lllll}
\textbf{1. } & $\alpha$ is even & $(-1)^{s}=+1$ &  & $\Rightarrow$ This case contributes with the positive sign in the
trace \tabularnewline
\end{tabular}

$ $

\begin{tabular}{lllll}
\textbf{2. } & $\alpha$ is odd \& $\beta$ is even & $(-1)^{s}=+1$ &  & \multirow{2}{*}{$\Rightarrow$ These two give zero contribution together.}\tabularnewline
\textbf{3.} & $\alpha$ is odd \& $\beta$ is odd & $(-1)^{s}=-1$ &  & \tabularnewline
\end{tabular}

As the case 2 and 3 add up to zero, we only  consider the case 1: 
\begin{equation}
\langle K \rangle=\frac{1}{2^N}\sum_{\alpha=0,2,\dots }^m\binom{m}{\alpha}\sum_{\beta=0}^{N-m}\binom{N-m}{\beta}=\frac{2^{m-1}2^{N-m}}{2^N}=\frac{1}{2}.
\end{equation}
So, we get that if $m$ is NOT divisible by 4, 
\begin{equation}
\langle \underset{m}{\underbrace{X\dots X}}Z\dots Z\rangle =\frac{1}{2}.
\end{equation}
This completes the proof.
\end{proof}

\subsection{B.3. Correlations for X measurements
on fully-connected three-uniform hypergraph states}

\begin{lemma} If every party makes a measurement in $X$ direction, then the expectations value is
\begin{equation}
\langle \underset{N}{\underbrace{XX\dots XX}}\rangle =\begin{cases}
\begin{array}[t]{cc}
0 & \mbox{\ensuremath{\mbox{if \ensuremath{N=0} mod \ensuremath{4}}}},\\
1  & \mbox{\ensuremath{\mbox{if \ensuremath{N=2} mod \ensuremath{4}}}}.
\end{array}\end{cases}
\end{equation}
\end{lemma} 
\begin{proof} In this proof we employ the  notation introduced in details in the proof of the lemma \ref{lemma3full}.
\begin{equation} \label{allX3unif}
\langle K \rangle:=\bra{H_3^N} XX \dots XX\ket{H_3^N}=\bra{+}^{\otimes N} \Big[ \prod_{e\in E} C_e \Big] XX\dots XX \Big[ \prod_{e\in E} C_e\Big]  \ket{+}^{\otimes N}.
\end{equation}
We use the identity (\ref{identity}), to regroup the phase gates on the right hand side of the expression (\ref{allX3unif}). Therefore, we count the multiplicity of the remaining phase gates:\\

\begin{tabular}{crlcc}
 & \#$\circledast\circledast\circledast$ & Each $C_{e},$ where $|e|=3$, occurs once and cancels with the one on right hand side.  &  & \tabularnewline
 & \#$\circledast\circledast$ & $=\binom{N-2}{1}$ is even  $\Rightarrow C_{\circledast\circledast}$ cancels& .  & \tabularnewline
 & \#$\circledast$ & $=\binom{N-1}{2}=\frac{(N-1)(N-2)}{2}$ is $\begin{cases}
\begin{array}[t]{cc}
\mbox{odd, if }N=0\ mod\ 4 & \Rightarrow C_{\circledast}\mbox{ remains.}\\
\mbox{even, if \ensuremath{N=2\ }\ensuremath{mod\ 4}} &\Rightarrow C_{\circledast}\mbox{ cancels.}
\end{array}\end{cases}$ &  & \tabularnewline
 & \#$\{\}$ & $=\binom{N}{3}=\frac{N(N-1)(N-2)}{2\cdot3}\ $ is even  $\Rightarrow$ global sign $GS$ is positive. &  & \tabularnewline
\end{tabular}\\

Therefore, we need to consider two cases to continue the derivation of the expression (\ref{allX3unif}): \\

\textbf{Case \#1:} If $N=0\ mod \; 4$, then
\begin{equation}
\langle K \rangle =\bra{+}^{\otimes N}XX\dots XX \prod_{\circledast\in E} C_{\circledast} \ket{+}^{\otimes N}= \bra{+}^{\otimes N}\prod_{\circledast\in E} C_{\circledast} \ket{+}^{\otimes N} =0.
\end{equation}

\textbf{Case \#2:} If $N=2\ mod \; 4$, then
\begin{equation}
\langle K \rangle =  \bra{+}^{\otimes N}XX\dots XX \ket{+}^{\otimes N}=\Big[\braket{+}{+}\Big] ^N=1.
\end{equation}
\end{proof}

\subsection{B.4. Correlations for Z measurements
on fully-connected three-uniform hypergraph states}

\begin{lemma}
If every party makes a measurement in $Z$ direction, the expectation value is zero. Therefore, we need to  show that 
\begin{equation}
\langle K \rangle :=\langle ZZ\dots ZZ\rangle=0
\end{equation}
\end{lemma}
\begin{proof}
\begin{equation}
\langle K \rangle =  \bra{+}^{\otimes N}\Big[\prod_{e \in E}C_e\Big]ZZ\dots ZZ\Big[\prod_{e \in E}C_e\Big] \ket{+}^{\otimes N}=\bra{+}^{\otimes N}ZZ\dots ZZ \ket{+}^{\otimes N}=\pm\frac{1}{2^N}	Tr(Z\dots ZZ)=0.
\end{equation}
\end{proof}

\section{Appendix C: Proof of Observation 5}

\noindent
{\bf Observation 5.}
{\it 
If we fix in the Bell operator $\BB_N$ in Eq.~(\ref{bell2}) the measurements 
to be $A=Z$ and $B=X$, then the $N$-qubit fully-connected three-uniform 
hypergraph state violates the classical bound, by an amount that grows 
exponentially with number of qubits, namely
\begin{align}
\left\langle \BB_N \right\rangle_C &\leq 2^{\left\lfloor N/2\right\rfloor} \quad  \mbox{ for local  HV models, and}
\nonumber
\\
\left\langle \BB_N \right\rangle_Q &\geq 2^{N-2}-\frac{1}{2} \quad  \mbox{ for the hypergraph state.}
\end{align}
}

The classical bound can be computed using the bound for Mermin inequality, i.e., $\mean{\BB^M_N}_C\leq 2^{\lfloor N/2 \rfloor}$. If $N$ is odd, our Bell inequality is exactly the Mermin inequality. If $N$ is even, 
$\mean{\BB_N}_C = \mean{A \cdot  \BB^M_{N-1} +B \cdot \tilde\BB^M_{N-1}}_C \leq 2 \mean{\BB^M_{N-1}}_C $, where $\tilde\BB$ denotes the inequality where $A$ and $B$ are exchanged, and the claim follows. The quantum value is computed directly. In particular, for odd $N$, we have
$
\left\langle \BB_N \right\rangle_Q=\sum_{k \;even}^N\binom{N}{k}\frac{1}{2}=2^{N-2}. 
$
If $N=0\ {\rm mod}\ 4$, then
$
\left\langle \BB_N \right\rangle_Q=\sum_{k\ even}^{N-1}\binom{N}{k}\frac{1}{2}-\left\langle X\dots X \right\rangle=2^{N-2}-{1}/{2}
$
and for $N=2\ mod\ 4$, we have
$
\left\langle \BB_N \right\rangle_Q=\sum_{k\ even}^{N-1}\binom{N}{k}\frac{1}{2}-\left\langle X\dots X \right\rangle=2^{N-2}+{1}/{2}.
$

Notice that, since Eq.~\eqref{bell2} in the main text is a full-correlation Bell inequality,
it can be written as a sum of probabilities minus a constant equal to 
the number of terms. With the proper normalization, such probabilities 
exactly correspond to the success probability of computing a Boolean 
function (cf. conclusive discussion and Ref.~\cite{Hoban11}).
%%%%%%%%%%%%%%%%%%%%%%%%%%%%%%%%%%%%

\section{Appendix D: Correlations for four-uniform hypergraph states}

%\subsection{D.1. Expectation values for fully-connected 
%four-uniform hypergraph states}

\begin{lemma}\label{lemma4full}
The following statements hold for  N-qubit, four-uniform hypergraph states:
\begin{enumerate}
\item  For the case $N=8k-2 \mbox{ or  }  \;  N=8k-1,\; \mbox{ or  } 8k $, we have:\\
$(i)$
\begin{equation}\label{third}
\langle \underset{m}{\underbrace{X\dots X}}Z\dots Z\rangle =\begin{cases}
\begin{array}[t]{cc}
+\frac{2^{\left\lfloor N/ 2\right\rfloor  -m}+1}{2^{\left\lfloor N/ 2\right\rfloor  - \left\lfloor m/ 2\right\rfloor }}  & \mbox{\ensuremath{\mbox{if \ensuremath{(m-1)=0} mod \ensuremath{4}}}},\\
-\frac{2^{\left\lfloor N/ 2\right\rfloor  -m}+1}{2^{\left\lfloor N/ 2\right\rfloor  - \left\lfloor m/ 2\right\rfloor }}  & \mbox{\ensuremath{\mbox{if \ensuremath{(m-1)=2} mod \ensuremath{4}}}}.
\end{array}\end{cases}
\end{equation}
$(ii)$ For $N=8k-1$, we have: 
\begin{equation}
\langle \underset{N}{\underbrace{XX\dots XX}}\rangle =-1. 
\end{equation} 
For $N=8k-2$ or $N=8k$, these correlations will not be needed.\\

\item For $N=4k+1$, we have: 
\begin{equation}
\langle \underset{m}{\underbrace{X\dots X}}Z\dots Z\rangle =\begin{cases}
\begin{array}[t]{cc}
+\frac{1}{2^{\left\lceil m/2\right\rceil }}, & \mbox{if \ensuremath{(m-1)=0 \ mod\ 4}},\\
-\frac{1}{2^{\left\lceil m/2\right\rceil }}, & \ \mbox{if \ensuremath{(m-1)=2 \ mod\ 4}},\\
\frac{1}{2^{\left\lfloor N/2\right\rfloor }}, & \mbox{if }\mbox{\ensuremath{m=N}}.
\end{array}\end{cases}
\end{equation}\\

\item For  $N=8k+2,\;   \mbox{ or  } 8k+4 $, we have for even $m$:\\
(i)
\begin{equation}
\langle \underset{m}{\underbrace{X\dots X}}Z\dots Z\rangle =\begin{cases}
\begin{array}[t]{cc}
+\frac{2^{m/2-1}}{2^{N/2 }} & \mbox{if \ensuremath{(N-m)=0\  mod\  4 }},\\
-\frac{2^{m/2-1}}{2^{N/2 }} & \mbox{if \ensuremath{(N-m)=2 \ mod\  4 }}.
\end{array}\end{cases}
\end{equation}\\
(ii) 
\begin{equation}
 \langle XX\dots XX\rangle =\frac{2^{\frac{N}{2}-1}+1}{2^{\frac{N}{2}}}.
\end{equation} \\
\item For $N=8k+3$,  $\langle \underset{m}{\underbrace{X\dots X}}\underset{N-m-1}{\underbrace{Z\dots Z}}\mathbbm1\rangle$ for even $m$ gives the same  exact result as the part 3, so we have:
\begin{equation}
\langle \underset{m}{\underbrace{X\dots X}}\underset{N-m-1}{\underbrace{Z\dots Z}}\mathbbm1\rangle=\begin{cases}
\begin{array}[t]{cc}
+\frac{2^{m/2-1}}{2^{M/2 }} & \mbox{if \ensuremath{(M-m)=0\  mod\  4 }},\\
-\frac{2^{m/2-1}}{2^{M/2 }} & \mbox{if \ensuremath{(M-m)=2 \ mod\  4 }},
\end{array}\end{cases}
\end{equation}
where $M=N-1.$
\end{enumerate}
\end{lemma}

\begin{proof} Each part of the theorem needs a separate consideration. note that the notation and machinery that is employed in the proof is based on the proof of the Lemma \ref{lemma3full}. Therefore, we advise the reader to become familiar with that one first. \\

\textbf{ Part 1. }   We consider the cases when $N=8k-2,  \; N=8k-1\; \mbox{ or  } 8k $ and odd $m$  together.  \\
We prove $(i)$ first:\\
\begin{equation}\label{part1_1}
\langle G_1 \rangle = \bra{H_4^N} \underset{m}{\underbrace{X\dots X}}Z\dots Z\ket{H_4^N}  =\bra{+}^{\otimes N}\Big(\prod_{e\in E}C_e \Big)XX\dots XZ\dots Z \Big(\prod_{e\in E}C_e \Big)  \ket{+}^{\otimes N}.
\end{equation}
We need to use the identity (\ref{identity}) to regroup all the phase gates on the right hand side of the expression (\ref{part1_1}). If each phase gate occurs even number of times, they give an identity, otherwise, they are used in the further calculations. We consider each case separately in the following table:

\begin{tabular}{rrllll}
 &  & \multicolumn{4}{l}{}\tabularnewline
\textbf{1.} & \#$\bigtriangleup\bigtriangleup\bigtriangleup$ & \multicolumn{4}{l}{$=\binom{m}{1}$ is odd $\Rightarrow C_{\bigtriangleup\bigtriangleup\bigtriangleup}$
remains. }\tabularnewline
 &  &  &  &  & \tabularnewline
\textbf{2.} & \#$\bigtriangleup\bigtriangleup$ & \multicolumn{4}{l}{$=\binom{m}{2}=\frac{m(m-1)}{2}$ is $\begin{cases}
\begin{array}[t]{cc}
\mbox{even, if (\ensuremath{m-1)=0}  mod 4 } & \Rightarrow  C_{\bigtriangleup\bigtriangleup}\mbox{ cancels.}\\
\mbox{odd, if (\ensuremath{m-1)=2}  mod 4 } & \Rightarrow  C_{\bigtriangleup\bigtriangleup}\mbox{ remains.}
\end{array}\end{cases}$ }\tabularnewline
 & \#$\circledast\bigtriangleup\bigtriangleup$ & \multicolumn{4}{l}{$=\binom{m-1}{1}$ is even $\Rightarrow\; C_{\circledast\bigtriangleup\bigtriangleup}$ cancels.}\tabularnewline
 &  &  &  &  & \tabularnewline
\textbf{3.} & \#$\bigtriangleup$ & \multicolumn{4}{l}{$=\binom{m}{3}=\frac{m(m-1)(m-2)}{2\cdot3}$ is $\begin{cases}
\begin{array}[t]{cc}
\mbox{even,  if (\ensuremath{m-1)=0}  mod 4 } &\Rightarrow C_{\bigtriangleup}\mbox{ cancels.}\\
\mbox{odd,  if (\ensuremath{m-1)=2}  mod 4 } & \Rightarrow C_{\bigtriangleup}\mbox{ remains.}
\end{array}\end{cases}$}\tabularnewline
 & \#$\circledast\circledast\bigtriangleup$ & \multicolumn{4}{l}{$=\binom{m-2}{1}$ is odd $\Rightarrow$ $C_{\circledast\circledast\bigtriangleup}$
remains.}\tabularnewline
 & \#$\circledast\bigtriangleup$ & \multicolumn{4}{l}{$=\binom{m-1}{2}=\frac{(m-1)(m-2)}{2}$ is $\begin{cases}
\begin{array}[t]{cc}
\mbox{even,   if (\ensuremath{m-1)=0}  mod 4 } &  \Rightarrow  C_{\circledast\bigtriangleup}\mbox{ cancels.}\\
\mbox{odd,  if (\ensuremath{m-1)=2}  mod 4 }  & \Rightarrow  C_{\circledast\bigtriangleup}\mbox{ remains.}
\end{array}\end{cases}$}\tabularnewline
 &  &  &  &  & \tabularnewline
\textbf{4.} 
 & \#$\circledast\circledast\circledast$ & \multicolumn{4}{l}{$=\binom{m-3}{1}$  is even $\Rightarrow\; C_{\circledast\circledast\circledast}$ cancels.}\tabularnewline
 & \#$\circledast\circledast$ & \multicolumn{4}{l}{$=\binom{m-2}{2}=\frac{(m-2)(m-3)}{2}$ is $\begin{cases}
\begin{array}[t]{cc}
\mbox{odd,   if (\ensuremath{m-1)=0}  mod 4 } & \Rightarrow C_{\circledast\circledast}\mbox{ remains.}\\
\mbox{even, if (\ensuremath{m-1)=2}  mod 4 } & \Rightarrow C_{\circledast\circledast}\mbox{ cancels.}
\end{array}\end{cases}$}\tabularnewline
 & \#$\circledast$ & \multicolumn{4}{l}{$=\binom{m-1}{3}=\frac{(m-1)(m-2)(m-3)}{2\cdot3}$ is even $\Rightarrow\; C_{\circledast}$ cancels. }\tabularnewline
 & \#$\{\}$ & \multicolumn{4}{l}{$=\binom{m}{4}=\frac{(m-1)(m-2)(m-3)(m-4)}{2\cdot3\cdot4}\ $ affects
global sign ($GL$) and will be discussed separately.}\tabularnewline
\end{tabular}
\begin{center}
\textbf{Table 1.}  Counting phase gates for a four-uniform case. $m$ is odd.
\end{center}

$\quad $\\
\textbf{Remark:} All four-qubit phase gates move with the multiplicity one to the right side and therefore, cancel out with the same phase gate on the right. The detailed reasoning was discussed in proof of the Lemma (\ref{lemma3full}) in the Appendix A. So, we skipped such scenarios in the Table 1.  

Now we consider  two cases of $(m-1)$ separately and for each case we fix  the global sign ($GL$) defined in the Table 1.\\

\textbf{Case \# 1:}  $(m-1)=0$ mod  4: 
\begin{equation}\label{part1_2}
\langle G_1 \rangle  =\pm\bra{+}^{\otimes N}  XX\dots XZ\dots Z \prod_{\forall \circledast, \forall\bigtriangleup } C_{\bigtriangleup \bigtriangleup \bigtriangleup} C_{\circledast \circledast \bigtriangleup}C_{\circledast\circledast} \ket{+}^{\otimes N}=\pm \frac{1}{2^N} Tr( \prod_{\forall \circledast, \forall\bigtriangleup } C_{\bigtriangleup \bigtriangleup \bigtriangleup} C_{\circledast \circledast \bigtriangleup}C_{\circledast\circledast}C_{\bigtriangleup}).
\end{equation}
\textbf{Remark:} We write the '$\pm$'  sign as we have not fixed the global sign yet.\\

To evaluate the trace of the given diagonal matrix, we need to find the difference between the  number of $+1$'s and  $-1$'s on the diagonal.  We write every row in the computational basis by enumerating it with the binary notation. Due to the symmetry of the problem, we assign the first $m$ columns to $X$ measurement  ($\circledast$) and the rest to, $Z$($\bigtriangleup$). For each  row, we denote by $\alpha$ the number of 1's in binary notation appearing in the first $m$ column and by $\beta$, the same on the rest. This notation is also adopted. See the proof of Lemma \ref{lemma3full} for more detailed explanation.\\

Considering the phase gates in  (\ref{part1_2}), the expression $(-1)^s$ defines whether in the given row the diagonal element is $+1$ or  $-1$, where : 
\begin{equation}
s:=\binom{\beta}{ 3}+\binom{\alpha}{2}\binom{\beta}{1}+\binom{\alpha}{2}+\binom{\beta}{1}=\frac{\beta (\beta-1)(\beta-2)}{2\cdot 3}+\frac{\alpha(\alpha -1)}{2}(\beta+1)+\beta.
\end{equation}
The sign of the diagonal element is determined at follows: 

\begin{tabular}{lllc}
\textbf{1. } & $\alpha$ is even \& $\beta$ is even:  & \multicolumn{2}{l}{if $\alpha=0$ mod 4 $\Rightarrow$ $(-1)^{s}=+1$}\tabularnewline
 &  & \multicolumn{2}{l}{if $\alpha=2 $ mod 4 $\Rightarrow$ $(-1)^{s}=-1$}\tabularnewline
 &  &  & \tabularnewline
\textbf{2. } & $\alpha$ is odd \& $\beta$ is even:  & \multicolumn{2}{l}{if $(\alpha-1)=0$ mod 4 $\Rightarrow$ $(-1)^{s}=+1$}\tabularnewline
 &  & \multicolumn{2}{l}{if $(\alpha-1)=2$ mod 4 $\Rightarrow$ $(-1)^{s}=-1$}\tabularnewline
 &  &  & \tabularnewline
\textbf{3.} &  (Any $\alpha$)  \& $\beta$ is odd:  & \multicolumn{2}{l}{if $(\beta-1)=0$ mod 4 $\Rightarrow$ $(-1)^{s}=-1$}\tabularnewline
 &  & \multicolumn{2}{l}{if $(\beta-1)=2$ mod 4 $\Rightarrow$ $(-1)^{s}=+1$}\tabularnewline
\end{tabular}\\

Having established the $\pm 1$ values for each row, we can sum them up to find the trace in (\ref{part1_2}). Here we use the identities (\ref{identity1}) and (\ref{identity2})  and afterwards the look-up Table 0 to insert the numerical values where necessary:

\begin{align} \label{part1_3}
\begin{split}
\langle G_1 \rangle  = &\pm\frac{1}{2^N}\Bigg[\sum^{N-m}_{\beta=0,2,4\dots }  \binom{N-m}{\beta}\bigg[\sum_{\alpha =0,4,\dots}^{m}\binom{m}{\alpha}+\binom{m}{\alpha +1}-\binom{m}{\alpha +2}-\binom{m}{\alpha +3}\bigg] \\
+&\sum_{\beta =1,5,\dots}  \bigg[-\binom{N-m}{\beta}\sum_\alpha \binom{m}{\alpha}+\binom{N-m}{\beta+2}\sum_\alpha \binom{m}{\alpha}\bigg]\Bigg]\\
= &\pm\frac{1}{2^N}\Bigg[\bigg[Re\Big[(1+i)^m \Big] +Im\Big[(1+i)^m\Big]\bigg]2^{N-m-1} + 2^m \sum^{N-m}_{\beta =1,5,\dots}  \bigg[-\binom{N-m}{\beta}+\binom{N-m}{\beta+2}\bigg]\Bigg]\\
=&\pm\frac{1}{2^N}\Bigg[\bigg[Re \Big[(1+i)^m \Big] +Im\Big[(1+i)^m\Big]\bigg]2^{N-m-1} - 2^m  Im\Big[(1+i)^{N-m}\Big]\Bigg] \equiv\pm\frac{1}{2^N} E.
\end{split}
\end{align}
We have to consider $N=8k-1$ and  $N=8k$ or $N=8k-2$ separately to continue the derivation of (\ref{part1_3}):\\

1.  For $N=8k-1$, using the values from the Table 0:
\begin{equation}
\langle G_1 \rangle=\pm \frac{1}{2^N} \Bigg[2^{\frac{m+1}{2}}2^{N-m-1}- 2^m  Im\Big[(1+i)^{N-m}\Big]\Bigg] =\pm\frac{2^{\frac{m+1}{2}}2^{N-m-1}+2^m2^{\frac{N-m}{2}}}{2^N}= \pm \frac{2^{\left\lfloor N/ 2\right\rfloor  -m}+1}{2^{\left\lfloor N/ 2\right\rfloor  - \left\lfloor m/ 2\right\rfloor }}.
\end{equation} \\

2.  For $N=8k$ or$N=8k-2$ , using the values from the Table 0:
\begin{equation}
\langle G_1 \rangle=\pm \frac{1}{2^N} \Bigg[2^{\frac{m+1}{2}}2^{N-m-1}- 2^m  Im\Big[(1+i)^{N-m}\Big]\Bigg] =\pm\frac{2^{\frac{m+1}{2}}2^{N-m-1}+2^m2^{\frac{N-m-1}{2}}}{2^N}= \pm \frac{2^{\left\lfloor N/ 2\right\rfloor  -m}+1}{2^{\left\lfloor N/ 2\right\rfloor  - \left\lfloor m/ 2\right\rfloor }}.
\end{equation}

Therefore,
\begin{equation}\label{part1_4}
 \langle G_1 \rangle = \pm \frac{2^{\left\lfloor N/ 2\right\rfloor  -m}+1}{2^{\left\lfloor N/ 2\right\rfloor  - \left\lfloor m/ 2\right\rfloor }}.
\end{equation}
Concerning the sign in (\ref{part1_4}), it is affected by  the product of two components: one from the case \textbf{4} from Table 1: $GL$ and the other by $E$ from  (\ref{part1_3}).  If $(m-1)=0$  mod 8,   the equation $E$ has a positive sign and also $GL=+1$. And  if $(m-1)=4$ mod 8, $E$  has a negative sign and $GL=-1$.  Therefore, in both cases  or equivalently, for $(m-1) =0\; mod\; 4$,
\begin{equation}
\langle G_1\rangle =+\frac{2^{\left\lfloor N/ 2\right\rfloor  -m}+1}{2^{\left\lfloor N/ 2\right\rfloor  - \left\lfloor m/ 2\right\rfloor }}.
\end{equation}

\textbf{Case \# 2:}  $(m-1)=2$ mod  4: 
\begin{align}\label{par2_1}
\begin{split}
\langle G_1 \rangle & =\pm\bra{+}^{\otimes N}  XX\dots XZ\dots Z \prod_{\forall \circledast, \forall\bigtriangleup } C_{\bigtriangleup \bigtriangleup \bigtriangleup} C_{\bigtriangleup \bigtriangleup } C_{\circledast \circledast \bigtriangleup}C_{\circledast\bigtriangleup}C_{\bigtriangleup } \ket{+}^{\otimes N}\\
& =\pm \frac{1}{2^N} Tr( \prod_{\forall \circledast, \forall\bigtriangleup }  C_{\bigtriangleup \bigtriangleup \bigtriangleup} C_{\bigtriangleup \bigtriangleup } C_{\circledast \circledast \bigtriangleup}C_{\circledast\bigtriangleup}).
\end{split}
\end{align}
In this case, we apply the same technique to determine the sign: 

\begin{equation}
s:=\binom{\beta}{ 3}+\binom{\beta}{ 2}+\binom{\alpha}{2}\binom{\beta}{1}+\binom{\alpha}{1}\binom{\beta}{1}=\frac{\beta (\beta-1)(\beta-2)}{2\cdot 3}+\frac{\beta(\beta-1)}{2}+\frac{\alpha(\alpha -1)}{2}\beta+\alpha\beta.
\end{equation}
The sign of $s$, is determined at follows: \\

\begin{tabular}{lllc}
\textbf{1. } & $\beta$ is even \&  any $\alpha$  & \multicolumn{2}{l}{if $\beta=0$ mod 4 $\Rightarrow$ $(-1)^{s}=+1$}\tabularnewline
 &  & \multicolumn{2}{l}{if $\beta=2$ mod 4 $\Rightarrow$ $(-1)^{s}=-1$}\tabularnewline
 &  &  & \tabularnewline
\textbf{2. } & $\beta$ is odd \& $\alpha$ is even:  & \multicolumn{2}{l}{if $\alpha=0$ mod 4 $\Rightarrow$ $(-1)^{s}=+1$}\tabularnewline
 &  & \multicolumn{2}{l}{if $\alpha=2$ mod 4 $\Rightarrow$ $(-1)^{s}=-1$}\tabularnewline
 &  &  & \tabularnewline
\textbf{3.} & $\beta$ is odd \& $\alpha$ is odd:  & \multicolumn{2}{l}{if $(\alpha-1)=0$ mod 4 $\Rightarrow$ $(-1)^{s}=-1$}\tabularnewline
 &  & \multicolumn{2}{l}{if $(\alpha-1)=2$ mod 4 $\Rightarrow$ $(-1)^{s}=+1$}\tabularnewline
\end{tabular}\\

\begin{align} \label{part2_2}
\begin{split}
\langle G_1\rangle  =& \pm\frac{1}{2^N} \Bigg[\sum_{\beta=0,4\dots } \bigg[\binom{N-m}{\beta}\sum_{\alpha =0}^{m}\binom{m}{\alpha}-\binom{N-m}{\beta +2}\sum_{\alpha =0}^{m}\binom{m}{\alpha}\bigg] \\
+ &\sum_{\beta=1,3,\dots } \binom{N-m}{\beta}\bigg[\sum_{\alpha =0,4,\dots}^{m}\binom{m}{\alpha}-\binom{m}{\alpha +1}-\binom{m}{\alpha +2}+\binom{m}{\alpha +3}\bigg]\Bigg]\\
= & \pm \frac{1}{2^N}\Bigg[2^m \sum_{\beta =0,4,\dots}  \bigg[\binom{N-m}{\beta}-\binom{N-m}{\beta+2}\bigg]+2^{N-m-1} \bigg[Re\Big[(1+i)^m \Big] -Im\Big[(1+i)^m\Big]\bigg]\Bigg]\\
= &  \pm \frac{1}{2^N}\Bigg[2^m  Re\Big[(1 +i)^{N-m}\Big] +2^{N-m-1} \bigg[Re\Big[(1+i)^m \Big] -Im\Big[(1+i)^m\Big]\bigg]\Bigg]\equiv \frac{1}{2^N}E.
\end{split} 
\end{align}

We have to consider $N=8k-1$ and  $N=8k$ or $N=8k-2$ separately to continue the derivation of (\ref{part1_3}):\\

1.  For $N=8k-1$, using the values from the Table 0:
\begin{equation}
\langle G_1 \rangle=\pm \frac{1}{2^N} \Bigg[ 2^m  Re\Big[(1+i)^{N-m}\Big]+2^{\frac{m+1}{2}}2^{N-m-1}\Bigg] =\pm\frac{2^m2^{\frac{N-m}{2}}+2^{\frac{m+1}{2}}2^{N-m-1}}{2^N}= \pm \frac{2^{\left\lfloor N/ 2\right\rfloor  -m}+1}{2^{\left\lfloor N/ 2\right\rfloor  - \left\lfloor m/ 2\right\rfloor }}.
\end{equation}
2.  For $N=8k$ or$N=8k-2$ , using the values from the Table 0:
\begin{equation}
\langle G_1 \rangle=\pm \frac{1}{2^N} \Bigg[ 2^m  Re\Big[(1+i)^{N-m}\Big]+2^{\frac{m+1}{2}}2^{N-m-1}\Bigg] =\pm\frac{2^m2^{\frac{N-m-1}{2}}+2^{\frac{m+1}{2}}2^{N-m-1}}{2^N}= = \pm \frac{2^{\left\lfloor N/ 2\right\rfloor  -m}+1}{2^{\left\lfloor N/ 2\right\rfloor  - \left\lfloor m/ 2\right\rfloor }}.
\end{equation}

Therefore,
\begin{equation}\label{part2_3}
 \langle G_1 \rangle = \pm \frac{2^{\left\lfloor N/ 2\right\rfloor  -m}+1}{2^{\left\lfloor N/ 2\right\rfloor  - \left\lfloor m/ 2\right\rfloor }}.
\end{equation}
Concerning the sign  in (\ref{part2_3}), it is affected by  the product of two components: one from the case \textbf{4} from Table 1: $GL$ and the other by $E$ in (\ref{part2_2}).\\
Hence, if $(m-3)=0$ mod 8,   $E$  has a negative sign and $GL=+1$. And  if $(m-3)=4$ mod 8,  $E$ has a positive sign and $GL=-1$.  Therefore, in both cases  or equivalently, for  $(m-1) =2\; mod\; 4$,
\begin{equation}\label{part2_4}
 \langle G_1 \rangle = - \frac{2^{\left\lfloor N/ 2\right\rfloor  -m}+1}{2^{\left\lfloor N/ 2\right\rfloor  - \left\lfloor m/ 2\right\rfloor }}.
\end{equation}
This completes the proof of part 1 (i).\\

\textbf{Part 1} $(ii)$ For $N=8k-1$,  show that
\begin{equation}
\langle G_2\rangle=\langle \underset{N}{\underbrace{XX\dots XX}}\rangle =-1.
\end{equation} 
Here as well we use the identity (\ref{identity}) to count the multiplicity of remaining phase  gates. Since all the measurements are in $X$ direction, we need to make a new table with the same notations as in the previous case:\\

\begin{tabular}{lrl}
 &  & \tabularnewline
$1.$& $\#\circledast\circledast\circledast\circledast$ &  every gate  occurs only once $\Rightarrow$  every $C_{e}$ cancels
with the $C_{e}$ on the right hand side.\tabularnewline
 &  & \tabularnewline
$2.$ & $\#\circledast\circledast\circledast$ & $\binom{8k-4}{1}$ is even  $\Rightarrow \; C_{\circledast\circledast\circledast}$ cancels.\tabularnewline
 &  & \tabularnewline
$3.$ & $\#\circledast\circledast$ &$ \binom{8k-3}{2}$ is even   $\Rightarrow \; C_{\circledast\circledast}$ cancels.\tabularnewline
 &  & \tabularnewline
$4.$ & $\#\circledast$ &  $\binom{8k-2}{3}$ is even  $\Rightarrow C_{\circledast}$ cancels.\tabularnewline
 &  & \tabularnewline
$5.$ & $\#\{\}$ &  $\binom{8k-1}{4}$ is odd $\Rightarrow$ we get a global negative sign. \tabularnewline
\end{tabular}
\begin{center}
\textbf{Table 2.} Counting phase gates for a four-uniform HG when each system is measured in $X$ direction.
\end{center}

$\quad $\\
Therefore, 
\begin{equation}
\langle G_2 \rangle=-\frac{1}{2^N}Tr\big(\mathbbm{1}\big)=-1.
\end{equation}
This finishes the proof of part 1.\\

\textbf{Part 2:} We show that, for $N=4k+1$: 
\begin{equation}
\langle G_3\rangle=\langle \underset{m}{\underbrace{X\dots X}}Z\dots Z\rangle =\begin{cases}
\begin{array}[t]{cc}
+\frac{1}{2^{\left\lceil m/2\right\rceil }}, & \mbox{if \ensuremath{(m-1)=0 \ mod\ 4}},\\
-\frac{1}{2^{\left\lceil m/2\right\rceil }}, & \ \mbox{if \ensuremath{(m-1)=2 \ mod\ 4}},\\
\frac{1}{2^{\left\lfloor N/2\right\rfloor }}, & \mbox{if }\mbox{\ensuremath{m=N}}.
\end{array}\end{cases}
\end{equation}

Since the number of systems measured in the $X$ direction is the same in this part as it was in part 1, we can use the results demonstrated in the Table 1. Therefore, we use  equation  (\ref{part1_3}) when $m-1=0$ mod 4  and (\ref{part2_2}), for $m-1=2$ mod 4. \\

\textbf{Case \# 1:}  $(m-1)=0$ mod  4: 
\begin{equation}\label{part3_1}
\langle G_3 \rangle=\pm\frac{1}{2^N}\Bigg[\bigg[Re \Big[(1+i)^m \Big] +Im\Big[(1+i)^m\Big]\bigg]2^{N-m-1} - 2^m  Im\Big[(1+i)^{N-m}\Big]\Bigg]=\pm\frac{1}{2^N} E.
\end{equation}
As $N=4k+1$,   we have that $\; N-m=4k+1-m=4k-(m-1)$,  which is divisible by $4$. Therefore  $Im\Big[(1+i)^{N-m}\Big]=0$. and  equation (\ref{part3_1}) reduces to:
\begin{equation}\label{part3_2}
\langle G_3 \rangle=\pm\frac{2^{N-m-1}}{2^N}\bigg[Re \Big[(1+i)^m \Big] +Im\Big[(1+i)^m\Big]\bigg]=\pm \frac{1}{2^{\left\lceil m/2\right\rceil }}.
\end{equation}
We need to fix the global sign $GL$ from the Table 1.  For this, we consider two cases. First, if $(m-1)=0$ mod 8, then $GL=+$ and  so is the sign $E$  in equation (\ref{part3_1}):
\begin{equation}
\langle G_3 \rangle=+\frac{1}{2^{\left\lceil m/2\right\rceil }}.
\end{equation}
Second, if  $(m-1)=4$ mod 8, then $GL=-$ and so is the sign of$E$ in  equation (\ref{part3_1}):
\begin{equation}
\langle G_3 \rangle=+\frac{1}{2^{\left\lceil m/2\right\rceil }}.
\end{equation}\\

\textbf{Case \# 2:}  For $(m-1)=2$ mod  4: 
\begin{equation}\label{part3_3}
\langle G_3 \rangle=  \pm \frac{1}{2^N}\Bigg[2^m  Re\Big[(1 +i)^{N-m}\Big] +2^{N-m-1} \bigg[Re\Big[(1+i)^m \Big] -Im\Big[(1+i)^m\Big]\bigg]\Bigg].
\end{equation}
 As $N=4k+1$,  we have $ N-m=4k+1-m=4k-(m-1)$, which is not divisible by $4$ but is an even number. Therefore,  $Re\Big[(1+i)^{N-m}\Big]=0$. So,  the equation (\ref{part3_3}) reduces to:
 \begin{equation}\label{part3_4}
\langle G_3 \rangle=  \pm \frac{2^{N-m-1}}{2^N} \bigg[Re\Big[(1+i)^m \Big] -Im\Big[(1+i)^m\Big]\bigg]\equiv \pm \frac{2^{N-m-1}}{2^N}  E.
\end{equation}
We need to fix the global sign $GL$ from the Table 1.  For this, we consider two cases. First, if $(m-3)=0$ mod 8, then the global sign is positive but the sign of $E$ in (\ref{part3_4}) is negative. Therefore, 
\begin{equation}
\langle G_3 \rangle=-\frac{1}{2^{\left\lceil m/2\right\rceil }}.
\end{equation}
Second, if $(m-3)=4$ mod 8, then the global sign is negative but the sign of $E$ in (\ref{part3_4})  is positive. Therefore, 
\begin{equation}
\langle G_3 \rangle=-\frac{1}{2^{\left\lceil m/2\right\rceil }}.
\end{equation}

\textbf{Case \# 3:}   $m=N$ resembles part 1 $(ii)$. 
The only difference comes in with the number of qubits we are currently working with:

\begin{tabular}{lrl}
 &  & \tabularnewline
$1$ & $\#\circledast\circledast\circledast\circledast$ & e every gate  occurs only once $\Rightarrow$ every $C_{e}$ cancels
with the $C_{e}$ on the right hand side.\tabularnewline
 &  & \tabularnewline
$2.$ & $\#\circledast\circledast\circledast$ & $\binom{4k-2}{1}$ is  even  $\Rightarrow \; C_{\circledast\circledast\circledast}$ cancels.\tabularnewline
 &  & \tabularnewline
$3.$ & $\#\circledast\circledast$ &  $\binom{4k-1}{2}$  is odd  $\Rightarrow C_{\circledast\circledast}$ remains. \tabularnewline
 &  & \tabularnewline
$4.$ & $\#\circledast$ & $\binom{4k}{3}$ is even   $\Rightarrow C_{\circledast}$
cancels.\tabularnewline
 &  & \tabularnewline
$5.$ & $\#\{\}$ & the global sign depends on $k$, as $\binom{4k+1}{4}=\frac{(4k+1)4k(4k-1)(4k-2)}{2\cdot 3\cdot 4}$  \tabularnewline
\end{tabular}\\
\begin{center}
\textbf{Table 3}.  Counting phase gates for a four-uniform HG when each system is measured in $X$ direction.
\end{center}

$\quad $\\
Back to the expectation value,
\begin{equation}
\langle G_3 \rangle=   \pm \frac{1}{2^N}Tr\bigg[\prod_{\forall \circledast}C_{\circledast\circledast}\bigg].
\end{equation}
So, we  have to count the difference between the amount of $+1$'s and $-1$'s on the diagonal. As we use exactly the same techniques before,  we will skip the detailed explanation.  The sign on the diagonal is:
\begin{equation}
(-1)^{\binom{\alpha}{2}}=(-1)^{\frac{\alpha(\alpha-1)}{2}}.
\end{equation} 
and it is straightforward to evaluate it for each value of $\alpha$. \\
\begin{equation}\label{aaa}
\langle G_3 \rangle=\pm\frac{1}{2^N}\Bigg [\sum_{\alpha =0,4,\dots}^{N}\binom{N}{\alpha}+\binom{N}{\alpha +1}-\binom{N}{\alpha +2}-\binom{N}{\alpha +3}\Bigg ] =\pm\frac{1}{2^N}\bigg[Re \Big[(1+i)^N \Big] +Im\Big[(1+i)^N\Big]\bigg]\equiv \pm\frac{1}{2^N}E.
\end{equation}\\

Keeping in mind that $N=4k+1$, the global sign from the Table 3 is positive for even $k$ and negative for odd. The the sign of $E$ in (\ref{aaa}) is positive if $k$ is even and negative, otherwise. Therefore,
\begin{equation}
\langle G_3 \rangle=\frac{1}{2^{\left\lfloor N/2\right\rfloor }}.
\end{equation}

This completes the proof of part 2.

\textbf{Part 3: }  We start with (ii).  We show that for  $N=8k+2,\;   \mbox{ or  } 8k+4 $:\\
(ii)
\begin{equation}
\langle G_4 \rangle= \langle XX\dots XX\rangle =\frac{2^{\frac{N}{2}-1}+1}{2^{\frac{N}{2}}}.
\end{equation} 

 Although the result seems identical, unfortunately, each case needs a separate treatment. The technique is similar to the previous proofs, though. We just  mind the number of qubits  we are working with:\\

For $N=8k+2$ we find the remaining phase gates  as follows: \\

\begin{tabular}{lrl}
 &  & \tabularnewline
$1$ & $\#\circledast\circledast\circledast\circledast$ & each gate  only once; thus,  every $C_{e}$ cancels
with the $C_{e}$ on the right hand side.\tabularnewline
 &  & \tabularnewline
$2.$ & $\#\circledast\circledast\circledast$ & $\binom{8k-1}{1}$ is odd  $\Rightarrow \; C_{\circledast\circledast\circledast}$ remains.\tabularnewline
 &  & \tabularnewline
$3.$ & $\#\circledast\circledast$ & $\binom{8k}{2}$ is  even  $\Rightarrow \; C_{\circledast\circledast}$ cancels.\tabularnewline
 &  & \tabularnewline
$4.$ & $\#\circledast$ &  $\binom{8k+1}{3}$ is even $\Rightarrow C_{\circledast }$ cancels.\tabularnewline
 &  & \tabularnewline
$5.$ & $\#\{\}$ &  $\binom{8k+2}{4}$ is even times $\Rightarrow$ we get a global positive sign. \tabularnewline
 &  & \tabularnewline
\end{tabular}
\begin{center}
\textbf{Table 4.} Counting phase gates for a four-uniform HG when each system is measured in $X$ direction.
\end{center}

$\quad $\\
Therefore,
\begin{equation}
\langle G_4 \rangle=\frac{1}{2^N}Tr\Big[
\prod_{\forall \circledast}C_{\circledast\circledast\circledast}\Big].
\end{equation}

We use $(-1)^s$ to define the sign of the diagonal element and $s=\binom{\alpha}{3}$. So, after considering all possible values of $\alpha$, it is directly obtained that 
\begin{align}
\begin{split}
\langle G_4 \rangle= \frac{1}{2^N} Tr(C_{\circledast\circledast\circledast}) & =\frac{1}{2^N}\bigg[\sum_{\alpha=0,2,\dots}^N\binom{N}{\alpha} + \sum_{\alpha=1,5,\dots}\Big[ \binom{N}{\alpha}-\binom{N}{\alpha+2}\Big] \bigg] \\ 
& =\frac{1}{2^N}\bigg[2^{N-1} +Im\big[(1+i)^N \big]\bigg]  =\frac{2^{\frac{N}{2}-1}+1}{2^{\frac{N}{2}}}.
\end{split}
\end{align}\\

For $N=8k+4$ we find the remaining phase gates  as follows: \\

\begin{tabular}{lrl}
 &  & \tabularnewline
$1.$ & $\#\circledast\circledast\circledast\circledast$ & each gate occurs only once $\Rightarrow$  every $C_{e}$ cancels with the $C_{e}$ on the right hand side.\tabularnewline
 &  & \tabularnewline
$2.$ & $\#\circledast\circledast\circledast$ & $ \binom{8k+1}{1}$ is  odd  $\Rightarrow \;  C_{\circledast\circledast\circledast}$ remains.\tabularnewline
 &  & \tabularnewline
$3.$ & $\#\circledast\circledast$ &  $\binom{8k+2}{2}$ is odd  $\Rightarrow\;  C_{\circledast\circledast}$ remains.\tabularnewline
 &  & \tabularnewline
$4. $ & $\#\circledast$ &  $\binom{8k+3}{3}$ is odd  $\Rightarrow \; =C_{\circledast}$ remains.\tabularnewline
 &  & \tabularnewline
$5. $ & $\#\{\}$ &  $\binom{8k+4}{4}$ is odd  $\Rightarrow$ we get a global negative  sign, $GL=-1$. \tabularnewline
 &  & \tabularnewline
\end{tabular}
\nopagebreak[9]
\begin{center}
\nopagebreak[9]
\textbf{Table 5.} Counting phase gates for a four-uniform case in all $X$ direction.
\end{center}

%$\quad $\\

Therefore,
\begin{equation}
\langle G_4\rangle  =-\frac{1}{2^N} Tr\Big[\prod_{\forall \circledast}C_{\circledast\circledast\circledast}C_{\circledast\circledast}C_{\circledast}\Big].
\vspace{-0.3cm}
\end{equation}
We use $(-1)^s$ to define the sign of the diagonal element and $s=\binom{\alpha}{3}+\binom{\alpha}{2}+\binom{\alpha}{1}$. So, after considering all possible values of $\alpha$, it is directly obtained that 
\begin{align}
\begin{split}
\langle G_4\rangle  =-\frac{1}{2^N} Tr(C_{\circledast\circledast\circledast}C_{\circledast\circledast}C_{\circledast})& =-\frac{1}{2^N}\bigg[\sum_{\alpha=0,4,\dots}^N\Big[\binom{N}{\alpha}-\binom{N}{\alpha+2}\Big] -\sum_{\alpha=1,3,\dots}^N \binom{N}{\alpha}\bigg]\\
&  =-\frac{1}{2^N}\bigg[-2^{N-1} +Re\big[(1+i)^N \big]\bigg]  =\frac{2^{N-1}+2^{N/2}}{2^N} =\frac{2^{\frac{N}{2}-1}+1}{2^{\frac{N}{2}}}.
\end{split}
\end{align}
This finishes the proof of part $(i)$.

(ii)  We need to show that
\begin{equation}
\langle G_4\rangle =\langle \underset{m}{\underbrace{X\dots X}}Z\dots Z\rangle =\begin{cases}
\begin{array}[t]{cc}
+\frac{2^{m/2-1}}{2^{N/2 }} & \mbox{if \ensuremath{(N-m)=0\  mod\  4 }},\\
-\frac{2^{m/2-1}}{2^{N/2 }} & \mbox{if \ensuremath{(N-m)=2 \ mod\  4 }}.
\end{array}\end{cases}
\end{equation}
Note that in this case $m$ is an even number. Therefore, we have to derive again from the scratch  how phase gates can be moved to the right hand side of the expression and for this we use the identity ($\ref{identity}$).

\begin{tabular}{rrllll}
 &  & \multicolumn{4}{l}{}\tabularnewline
\textbf{1.} & \#$\bigtriangleup\bigtriangleup\bigtriangleup$ & \multicolumn{4}{l}{$=\binom{m}{1}$ is even $\Rightarrow\;  C_{\bigtriangleup\bigtriangleup\bigtriangleup}$ cancels.}\tabularnewline
 &  &  &  &  & \tabularnewline
\textbf{2.} & \#$\bigtriangleup\bigtriangleup$ & \multicolumn{4}{l}{$=\binom{m}{2}=\frac{m(m-1)}{2}$ is $\begin{cases}
\begin{array}[t]{cc}
\mbox{even, if \ensuremath{m=0}  mod 4} & \Rightarrow  C_{\bigtriangleup\bigtriangleup}\mbox{ cancels.}\\
\mbox{odd, if \ensuremath{m=2} mod 4} & \Rightarrow  C_{\bigtriangleup\bigtriangleup}\mbox{ remains.}
\end{array}\end{cases}$ }\tabularnewline
 & \#$\circledast\bigtriangleup\bigtriangleup$ & \multicolumn{4}{l}{$=\binom{m-1}{1}$ is odd $\Rightarrow C_{\circledast\bigtriangleup\bigtriangleup}$ remains.}\tabularnewline
 &  &  &  &  & \tabularnewline
\textbf{3.} & \#$\bigtriangleup$ & \multicolumn{4}{l}{$=\binom{m}{3}=\frac{m(m-1)(m-2)}{2\cdot3}$ is even $\Rightarrow\; C_{\bigtriangleup}$ cancels.}\tabularnewline
 & \#$\circledast\circledast\bigtriangleup$ & \multicolumn{4}{l}{$=\binom{m-2}{1}$ is even $\Rightarrow\; C_{\circledast\circledast\bigtriangleup} $ cancels.}\tabularnewline
 & \#$\circledast\bigtriangleup$ & \multicolumn{4}{l}{$=\binom{m-1}{2}=\frac{(m-1)(m-2)}{2}$ is $\begin{cases}
\begin{array}[t]{cc}
\mbox{odd, if \ensuremath{m=0}  mod 4} & \Rightarrow  C_{\circledast\bigtriangleup}\mbox{ remains.}\\
\mbox{even,  if \ensuremath{m=2}  mod 4} & \Rightarrow  C_{\circledast\bigtriangleup}\mbox{ cancels.}
\end{array}\end{cases}$}\tabularnewline
 &  &  &  &  & \tabularnewline
\textbf{4.} & \#$\circledast\circledast\circledast$ & \multicolumn{4}{l}{$=\binom{m-3}{1}$  is odd $\Rightarrow \; C_{\circledast\circledast\circledast}$ remains. }\tabularnewline
 & \#$\circledast\circledast$ & \multicolumn{4}{l}{$=\binom{m-2}{2}=\frac{(m-2)(m-3)}{2}$ is $\begin{cases}
\begin{array}[t]{cc}
\mbox{odd, if \ensuremath{m=0}  mod 4} & \Rightarrow C_{\circledast\circledast}\mbox{ remains.}\\
\mbox{even, if \ensuremath{m=2}  mod 4} &  \Rightarrow C_{\circledast\circledast}\mbox{ cancels.}
\end{array}\end{cases}$}\tabularnewline
 & \#$\circledast$ & \multicolumn{4}{l}{$=\binom{m-1}{3}=\frac{(m-1)(m-2)(m-3)}{2\cdot3}$ is $\begin{cases}
\begin{array}[t]{cc}
\mbox{odd,  if \ensuremath{m=0}  mod 4} & \Rightarrow C_{\circledast}\mbox{ remains.}\\
\mbox{even,  if \ensuremath{m=2}  mod 4} & \Rightarrow C_{\circledast}\mbox{ cancels.}\end{array}\end{cases}$ }\tabularnewline
 & \#$\{\}$ & \multicolumn{4}{l}{$=\binom{m}{4}=\frac{m(m-1)(m-2)(m-3)}{2\cdot3\cdot4}\ $ affects
the global sign ($GL$) and will be discussed separately.}\tabularnewline
\end{tabular}
\nopagebreak[9]
\begin{center}
\textbf{Table 6.} Counting phase gates for a four-uniform HG state, for even $m$.
\end{center}

$\;$\\
\textbf{Remark:} Similarly to previous proofs the four-qubit phase gates cancel out. Therefore, we directly skip the discussion about them.

We need to consider two cases, when $m=0$ mod 4 and $m=2$ mod 4 for  each  $N=8k+2$ and $8k+4$ separately:\\

\textbf{Case \# 1:} If  $m=0$ mod  4:\\
\begin{align}\label{part4_1}
\begin{split}
\langle G_4 \rangle  & =\pm\bra{+}^{\otimes N}  XX\dots XZ\dots Z \prod_{\forall \circledast, \forall\bigtriangleup } C_{\circledast \bigtriangleup \bigtriangleup}C_{\circledast\bigtriangleup}  C_{\circledast \circledast \circledast } C_{\circledast \circledast }C_{\circledast} \ket{+}^{\otimes N}\\
& =\pm \frac{1}{2^N} Tr\Big[ \prod_{\forall \circledast, \forall\bigtriangleup } C_{\circledast \bigtriangleup \bigtriangleup}C_{\circledast\bigtriangleup}  C_{\circledast \circledast \circledast } C_{\circledast \circledast }C_{\circledast}C_{\bigtriangleup}  \Big].
\end{split}
\end{align}

We use $(-1)^s$ to define the sign of the diagonal element and $s=\binom{\alpha}{1}\binom{\beta}{2}+\binom{\alpha}{1}\binom{\beta}{1}+\binom{\alpha}{3}+\binom{\alpha}{2}+\binom{\alpha}{1}+\binom{\beta}{1}$.  If $s$ is even, the value on the diagonal is $+1$ and $-1$, otherwise. We consider all possible values of $\alpha$ and $\beta$:\\

\begin{tabular}{lllc}
\textbf{1. } & $\alpha$ is even \& $\beta$ is even & \multicolumn{2}{l}{if $\alpha=0$ mod 4 $\Rightarrow$ $(-1)^{s}=+1$}\tabularnewline
 &  & \multicolumn{2}{l}{if $\alpha=2$ mod 4 $\Rightarrow$ $(-1)^{s}=-1$}\tabularnewline
 &  &  & \tabularnewline
\textbf{2. } & $\alpha$ is even \& $\beta$ is odd:  & \multicolumn{2}{l}{if $\alpha=0$ mod 4 $\Rightarrow$ $(-1)^{s}=-1$}\tabularnewline
 &  & \multicolumn{2}{l}{if $\alpha=2$ mod 4 $\Rightarrow$ $(-1)^{s}=+1$}\tabularnewline
\end{tabular}\\

From here one can easily spot that for even $\alpha$, there is equal number of $+1$ and $-	1$ on the diagonal. So, they do not contribute in the calculations. We now consider the odd $\alpha$:\\

\begin{tabular}{lllc}
\textbf{3. } & $\alpha$ is odd \& $\beta$ is even & \multicolumn{2}{l}{if $\beta=0$ mod 4 $\Rightarrow$ $(-1)^{s}=-1$}\tabularnewline
 &  & \multicolumn{2}{l}{if $\beta=2$ mod 4 $\Rightarrow$ $(-1)^{s}=+1$}\tabularnewline
 &  &  & \tabularnewline
\textbf{4. } & $\alpha$ is odd \& $\beta$ is odd:  & \multicolumn{2}{l}{if $(\beta-1)=0$ mod 4 $\Rightarrow$ $(-1)^{s}=-1$}\tabularnewline
 &  & \multicolumn{2}{l}{if $(\beta-1)=2$ mod 4 $\Rightarrow$ $(-1)^{s}=+1$}\tabularnewline
\end{tabular}\\

We now continue calculation of the trace from ($\ref{part4_1}$): 
\begin{align}{\label{part4_2}}
\begin{split}
\langle G_4 \rangle & =\pm \frac{1}{2^N} \sum_{\alpha=1,3,5\dots }  \binom{m}{\alpha} \bigg[\sum_{\beta =0,4,\dots}^{N-m}-\binom{N-m}{\beta}-\binom{N-m}{\beta +1}+\binom{N-m}{\beta +2}+\binom{N-m}{\beta +3}\bigg] \\
& = \pm \frac{2^{m-1}}{2^N} \bigg[-Re \Big[(1+i)^{N-m} \Big] - Im\Big[(1+i)^{N-m}\Big]\bigg]=\pm \frac{2^{m-1}}{2^N} \bigg( \mp2^{\frac{N-m}{2}} \bigg).
\end{split}
\end{align}

We have to take care of the sign which appears from the product of the sign of the sum of real and imaginary part in ($\ref{part4_2}$) and global sign $(GL)$, which we defined while deriving the remaining phase gates.  If $m$ is divisible by $8$, $GL=+1$ and since we are in $N=8k+2$ case,  $(N-m=8k+2-m)-2$ mod $8$ and therefore: 
\begin{equation}
\langle G_4 \rangle= \frac{2^{m-1}}{2^N} \bigg( -2^{\frac{N-m}{2}} \bigg) = -\frac{2^{m/2-1}}{2^{N/2 }}.
\end{equation}
If $m$ is not divisible by $8$,  the global sign $GL=-1$, and the sum of real and imaginary part also contribute with a negative sign. Thus,
\begin{equation}
\langle G_4 \rangle=  -\frac{2^{m-1}}{2^N} \bigg( -(-2^{\frac{N-m}{2}}) \bigg) = -\frac{2^{m/2-1}}{2^{N/2 }}.
\end{equation}\\

Since the $N=8k+4$ case is identical, we only have to mind the sign of the sum of the real and imaginary part. Here as well we consider two cases: if $m$ is divisible by $8$, then the global sign $GL=+1$, and the sign of the sum of real and imaginary part is $"-"$. Therefore,
\begin{equation}
\langle G_4 \rangle=\frac{2^{m-1}}{2^N} \bigg( -(-2^{\frac{N-m}{2}}) \bigg) = +\frac{2^{m/2-1}}{2^{N/2 }}.
\end{equation}
And if $m$ is not divisible by $8$, $GL=-1$, and the sign of real and imaginary part is $"+"$. Therefore,
\begin{equation}
\langle G_4 \rangle= -\frac{2^{m-1}}{2^N} \bigg( -(+2^{\frac{N-m}{2}}) \bigg) = +\frac{2^{m/2-1}}{2^{N/2 }}.
\end{equation}\\

\textbf{Case \# 2:} If  $m=2$ mod  4:
\begin{align}\label{part2ii}
\begin{split}
\langle G_4\rangle & =\pm\bra{+}^{\otimes N}  XX\dots XZ\dots Z \prod_{\forall \circledast, \forall\bigtriangleup } C_{\circledast \bigtriangleup \bigtriangleup} C_{\circledast \circledast \circledast } C_{\bigtriangleup \bigtriangleup } \ket{+}^{\otimes N}\\
& =\pm \frac{1}{2^N} Tr( \prod_{\forall \circledast, \forall\bigtriangleup } C_{\circledast \bigtriangleup \bigtriangleup} C_{\circledast \circledast \circledast } C_{\bigtriangleup \bigtriangleup } C_{ \bigtriangleup }   ).
\end{split}
\end{align}\\
We use $(-1)^s$ to define the sign of the diagonal element and $s=\binom{\alpha}{1}\binom{\beta}{2}+\binom{\alpha}{3}+\binom{\beta}{2}+\binom{\beta}{1}$.  If $s$ is even, the value on the diagonal is $+1$ and $-1$, otherwise. We consider all possible values of $\alpha$ and $\beta$:\\

Considering the terms from odd $\alpha$:\\

\begin{tabular}{lllc}
\textbf{1. } & $\alpha$ is odd \& $\beta$ is even & \multicolumn{2}{l}{if $(\alpha-1)=0$ mod 4 $\Rightarrow$ $(-1)^{s}=+1$}\tabularnewline
 &  & \multicolumn{2}{l}{if $(\alpha-1)=2$ mod 4 $\Rightarrow$ $(-1)^{s}=-1$}\tabularnewline
 &  &  & \tabularnewline
\textbf{2. } & $\alpha$ is odd \& $\beta$ is odd:  & \multicolumn{2}{l}{if $(\alpha-1)=0$ mod 4 $\Rightarrow$ $(-1)^{s}=-1$}\tabularnewline
 &  & \multicolumn{2}{l}{if $(\alpha-1)=2$ mod 4 $\Rightarrow$ $(-1)^{s}=+1$}\tabularnewline
\end{tabular}

It is easy to see that these cases adds up to $0$. \\

\begin{tabular}{lllc}
\textbf{3. } & $\alpha$ is even \& $\beta$ is even & \multicolumn{2}{l}{if $\beta=0$ mod 4 $\Rightarrow$ $(-1)^{s}=+1$}\tabularnewline
 &  & \multicolumn{2}{l}{if $\beta=2$ mod 4 $\Rightarrow$ $(-1)^{s}=-1$}\tabularnewline
 &  &  & \tabularnewline
\textbf{4. } & $\alpha$ is even \& $\beta$ is odd:  & \multicolumn{2}{l}{if $(\beta-1)=0$ mod 4 $\Rightarrow$ $(-1)^{s}=-1$}\tabularnewline
 &  & \multicolumn{2}{l}{if $(\beta-1)=2$ mod 4 $\Rightarrow$ $(-1)^{s}=+1$}\tabularnewline
\end{tabular}\\

Therefore, 
\begin{align}
\begin{split}
\langle G_4\rangle=\pm\frac{1}{2^N}\sum_{\alpha=0,2,4\dots }  \binom{m}{\alpha} \bigg[\sum_{\beta =0,4,\dots}^{N-m}\binom{N-m}{\beta} & -\binom{N-m}{\beta +1}-\binom{N-m}{\beta +2}+\binom{N-m}{\beta +3}\bigg]  \\
=\pm 2^{m-1}  \bigg[Re \Big[(1+i)^{N-m} & \Big] - Im\Big[(1+i)^{N-m}\Big]\bigg]=\pm\frac{2^{m/2-1}}{2^{N/2 }}.
\end{split}
\end{align}

To fix the sign, we need to first consider $N=8k+2$,  and $(m-2)=0$  mod $8$. Then the global sign $GL=+1$ and type of $N-m$ also yields a positive sign. But if $(m-2)=4$ mod  $8$, global sign in negative and the $N-m$ also yields the negative sign. So, \\
\begin{equation}
\langle G_4 \rangle=\frac{2^{m/2-1}}{2^{N/2 }}.
\end{equation}
 $N=8k+4$ case is identical to $N=8k+2$, therefore we will just state the result. For $N=8k+4$
 \begin{equation}
\langle G_4 \rangle=-\frac{2^{m/2-1}}{2^{N/2 }}.
\end{equation}
To sum up,
\begin{equation}
\langle G_4\rangle  =\begin{cases}
\begin{array}[t]{cc}
+\frac{2^{m/2-1}}{2^{N/2 }} & \mbox{if \ensuremath{(N-m)=0\  mod\  4 }},\\
-\frac{2^{m/2-1}}{2^{N/2 }} & \mbox{if \ensuremath{(N-m)=2 \ mod\  4 }}.
\end{array}\end{cases}
\end{equation}\\
This finishes the proof of part 3.\\

\textbf{Part 4: }  We show that  for $N=8k+3$,  $\langle \underset{m}{\underbrace{X\dots X}}\underset{N-m-1}{\underbrace{Z\dots Z}\mathbbm1}\rangle$ for even $m$ gives the same exact result as  the part 3 (i).\\

We tackle  the problem as follows: We make a measurement on one of the qubits in $Z$ direction and  depending on the measurement outcome, we obtain the new $\ket{H_{4_{new}}^{M}}$  state, where $M:=N-1$. Then we consider the expectation values for the all possible measurement outcomes and $\ket{H_{4_{new}}^{M}}$. From that we conclude the statement in the  part 4. \\

Initial HG state $\ket{H_4^N}$ can be written in the following form as well:
\begin{equation}
\ket{H_4^N}=\prod_e C_e\ket{+}^{\otimes N}=\prod_{e',e''}C_{e'}C_{e''}\ket{+}^{\otimes N}=\prod_{e',e''}\big[\mathbbm{1} _{e'\backslash N}\ket{0}_N\bra{0}_N + C_{e'\backslash N}\ket{1}_N\bra{1}_N \big]C_{e''}\ket{+}^{\otimes N},
\end{equation}
where $e'$ represents the gates containing the last qubit, $N$ and $e''$ represents the ones which do not contain $N^{th}$ qubit. And $e=e'+e''$. Then if one makes a measurement in $Z$ basis on the last qubit and obtains outcome $+$, 
\begin{align}
\begin{split}
\ket{H_{4_{new}}^{M+}}=\braket{0_N}{H_4^N}&=\prod_{e',e''}\big[\bra{0}_N\mathbbm{1} _{e'\backslash N}\ket{0}_N\bra{0}_N+\bra{0}_NC_{e'\backslash N}\ket{1}_N\bra{1}_N \big]C_{e''}\ket{+}^{\otimes N}\\
&=\prod_{e',e''}\mathbbm{1} _{e'\backslash N}\bra{0}_NC_{e''}\ket{+}^{\otimes N}=\bra{0}_N\prod_{e',e''}\mathbbm{1} _{e'\backslash N}C_{e''}\ket{+}^{\otimes M}\big(\ket{0}_N+\ket{1}_N\big)\\
&=\prod_{e',e''}\mathbbm{1} _{e'\backslash N}C_{e''}\ket{+}^{\otimes M}=\prod_{e''}C_{e''}\ket{+}^{\otimes M}.
\end{split}.
\end{align}
So, $+$ outcome after measuring in $Z$ direction leaves us with $\ket{H_{4_{new}}^{M+}}$, which is precisely four uniform $M$-qubit HG state. Now, let us see what is the remaining state if one gets $-$ as an outcome result:
\begin{align}
\begin{split}
\ket{H_{4_{new}}^{M-}}=\braket{1_N}{H_4^N}&=\prod_{e',e''}\big[\bra{1}_N\mathbbm{1} _{e'\backslash N}\ket{0}_N\bra{0}_N+\bra{1}_NC_{e'\backslash N}\ket{1}_N\bra{1}_N \big]C_{e''}\ket{+}^{\otimes N}\\
&=\prod_{e',e''}C_{e'\backslash N}\bra{1}_N C_{e''}\ket{+}^{\otimes N}= \bra{1}_N  \prod_{e',e''}C_{e'\backslash N}C_{e''}\ket{+}^{\otimes M}\big(\ket{0}_N+\ket{1}_N\big)\\
&=\prod_{e',e''}C_{e'\backslash N}C_{e''}\ket{+}^{\otimes M}.
\end{split}
\end{align}
So, $-$ outcome after measuring in $Z$ direction leaves us with $\ket{H_{4_{new}}^{M-}}$, which is precisely a symmetric $M$-qubit HG state with all possible edges of cardinality four and three. We will call such HG state  a three- and four-uniform HG state.\\

Therefore, problem boils down to showing that,\\
$(i)$ If the measurement outcome is $+$, we get the $M=8k+2$ four-uniform HG state and the correlations are given in part 3.\\
$(ii)$ If the measurement outcome is $-$, we get $M=8k+2$  three- and fouruniform HG state and the following holds:
\begin{equation}
\langle G_5^-\rangle=\bra{H_{4_{new}}^{M-}} \underset{m}{\underbrace{X\dots X}}Z\dots Z\ket{H_{4_{new}}^{M-}}=\begin{cases}
\begin{array}[t]{cc}
-\frac{2^{m/2-1}}{2^{M/2 }} & \mbox{if \ensuremath{(M-m)=0\  mod\  4 }},\\
+\frac{2^{m/2-1}}{2^{M/2 }} & \mbox{if \ensuremath{(M-m)=2 \ mod\  4 }}.
\end{array}\end{cases}
\end{equation}

$ (i)$  $\ket{H_{4_{new}}^{M+}}$, where $M=8k+2$ was already considered in part 3.\\

 (ii). For $\ket{H_{4_{new}}^{M-}}$,
\begin{equation}
\langle G_5^-\rangle=\bra{+}^{\otimes M}\Big[\prod_{e',e''\in E} C_{e' }C_{e''}\Big]\underset{m}{\underbrace{X\dots X}}Z\dots Z\Big[\prod_{e',e''\in E} C_{e }C_{e''}\Big]\ket{+}^{\otimes M}.
\end{equation}
Before, we treated three- and four-uniform cases separately. Now, we just need to put them together. \\

\textbf{Case \# 1:} If  $m=0$ mod  4:\\
Then from equations (\ref{appendixAeq6}) and  (\ref{part4_1}), we can directly write down that
\begin{equation}
\langle G_5^-\rangle=\frac{1}{2^M}Tr\bigg[\prod_{\forall \circledast, \forall\bigtriangleup } C_{\circledast \bigtriangleup \bigtriangleup} C_{\circledast \circledast \circledast } C_{\circledast \circledast } C_{\bigtriangleup}\bigg]. 
\end{equation}

We check the sign of each term on the diagonal by $(-1)^s$, where $s=\binom{\alpha}{1}\binom{\beta}{2}+\binom{\alpha}{3}+\binom{\alpha}{2}+\binom{\beta}{1}$. For this we need to consider each value of $\alpha$ and $\beta$ separately.\\

\begin{tabular}{lllll}
\multirow{2}{*}{\textbf{1. }} & \multirow{2}{*}{$\alpha$ is even \& $\beta$ is even} & $\alpha=0$ mod 4 $\Rightarrow$ $+1$ &  & \multirow{5}{*}{$\Rightarrow$ These two give zero contribution together.}\tabularnewline
 &  & $\alpha=2$ mod 4 $\Rightarrow$ $-1$ &  & \tabularnewline
 &  &  &  & \tabularnewline
\multirow{2}{*}{\textbf{2.}} & \multirow{2}{*}{$\alpha$ is even \& $\beta$ is odd} & $\alpha=0$ mod 4 $\Rightarrow$ $-1$ &  & \tabularnewline
 &  & $\alpha=2$ mod 4 $\Rightarrow$ $+1$ &  & \tabularnewline
\end{tabular}\\
\\

\begin{tabular}{lllll}
\multirow{2}{*}{\textbf{3. }} & \multirow{2}{*}{$\alpha$ is odd \& $\beta$ is even} & $\beta=0$ mod 4 $\Rightarrow$ $+1$ &  & \multirow{5}{*}{}\tabularnewline
 &  & $\beta=2$ mod 4 $\Rightarrow$ $-1$ &  & \tabularnewline
 &  &  &  & \tabularnewline
\multirow{2}{*}{\textbf{4.}} & \multirow{2}{*}{$\alpha$ is odd \& $\beta$ is odd} & $\beta-1=0$ mod 4 $\Rightarrow$ $-1$&  & \tabularnewline
 &  & $\beta-1=2$ mod 4 $\Rightarrow$ $+1$ &  & \tabularnewline
\end{tabular}\\

\begin{align}\label{part5_1}
\begin{split}
\langle G_5^- \rangle & =\pm\frac{1}{2^{M}}\sum_{\alpha\ odd}^{m}\binom{m}{\alpha}\bigg[\sum_{\beta=0,4}^{M-m}\binom{M-m}{\beta}-\binom{M-m}{\beta+1}-\binom{M-m}{\beta+2}+\binom{M-m}{\beta+3}\bigg]\\
 &=\pm\frac{2^{m-1}}{2^{M}}\bigg[Re(1+i)^{M-m}-Im(1+i)^{M-m}\bigg]=\pm\frac{2^{\frac{m}{2}-1}}{2^{M/2}}.
\end{split}
\end{align}

If $m=0$ mod 8,  real and imaginary part in  (\ref{part5_1}) has a negative sign and the global sign coming from  Table 6, $GL$ is positive. Note from equation (\ref{3unifGL}) that  three uniform gate moving does not introduce any global signs. And if $m=4$ mod 8,  real and imaginary part in  (\ref{part5_1}) has a positive sign and the global sign coming from  Table 6, $GL$ is negative.Therefore, 
\begin{equation}
\langle G_5^+ \rangle =-\frac{2^{\frac{m}{2}-1}}{2^{M/2}}.
\end{equation}

\textbf{Case \# 2:} If  $m=2$ mod  4:\\
Then from equations (\ref{appendixAeq8}) and  (\ref{part2ii}), we can directly write down that
 \begin{equation}
 \langle G_5^-\rangle=\pm\frac{1}{2^{M}}Tr\bigg[\prod_{\forall \circledast, \forall\bigtriangleup } C_{\circledast\triangle\triangle}C_{\circledast\triangle}C_{\circledast\circledast\circledast}C_{\triangle\triangle}\bigg].
 \end{equation}
 We check the sign of each term on the diagonal by $(-1)^s$, where $s=\binom{\alpha}{1}\binom{\beta}{2}+\binom{\alpha}{3}+\binom{\beta}{2}+\binom{\alpha}{1}\binom{\beta}{1}$. For this we need to consider each value of $\alpha$ and $\beta$ separately.\\

\begin{tabular}{lllll}
\multirow{2}{*}{\textbf{1. }} & \multirow{2}{*}{$\alpha$ is even \& $\beta$ is even} & $\beta=0$ mod 4 $\Rightarrow$ $+1$ &  & \multirow{5}{*}{}\tabularnewline
 &  & $\beta=2$ mod 4 $\Rightarrow$ $-1$ &  & \tabularnewline
 &  &  &  & \tabularnewline
\multirow{2}{*}{\textbf{2.}} & \multirow{2}{*}{$\alpha$ is even \& $\beta$ is odd} & $\beta-1=0$ mod 4 $\Rightarrow$ $+1$ &  & \tabularnewline
 &  & $\beta-1=2$ mod 4 $\Rightarrow$ $-1$ &  & \tabularnewline
\end{tabular}\\
$\quad$\\

\begin{tabular}{lllll}
\multirow{2}{*}{\textbf{3. }} & \multirow{2}{*}{$\alpha$ is odd \& $\beta$ is even} &  & $\alpha-1=0$ mod 4 $\Rightarrow$ $+1$ & \multirow{4}{*}{$\Rightarrow$ These two give zero contribution together.}\tabularnewline
 &  &  & $\alpha-1=2$ mod 4 $\Rightarrow$ $-1$ & \tabularnewline
\multirow{2}{*}{\textbf{4.}} & \multirow{2}{*}{$\alpha$ is odd \& $\beta$ is odd} &  & $\alpha-1=0$ mod 4 $\Rightarrow$ $-1$ & \tabularnewline
 &  &  & $\alpha-1=2$ mod 4 $\Rightarrow$ $+1$ & \tabularnewline
\end{tabular} 

\begin{align}\label{part5_2}
\begin{split}
\langle G_{5}^-\rangle & = \pm\frac{1}{2^{M}}\sum_{\alpha\ even}^{m}\binom{m}{\alpha}\bigg[\sum_{\beta=0,4}^{M-m}\binom{M-m}{\beta}+\binom{M-m}{\beta+1}-\binom{M-m}{\beta+2}-\binom{M-m}{\beta+3}\bigg]\\
& = \pm\frac{1}{2^{M}} 2^{m-1}\bigg[Re(1+i)^{M-m}+Im(1+i)^{M-m}\bigg]=\pm\frac{2^{\frac{m}{2}-1}}{2^{M/2}}.
\end{split}
\end{align}

If $m-2=0$ mod 8,  real and imaginary part in  (\ref{part5_2}) has a positive sign and the global sign coming from  Table 6, $GL$ is positive. Note from equation (\ref{3unifGL}) that the three uniform gate moving does not introduce any global signs. And if $m-2=4$ mod 8,  real and imaginary part in  (\ref{part5_2}) has a negative sign and the global sign coming from  Table 6, $GL$ is negative.Therefore, 

\begin{equation}
\langle G_5^- \rangle =+\frac{2^{\frac{m}{2}-1}}{2^{M/2}}.
\end{equation}
Finally, we can put everything together. Since one can observe that $\langle G_5^- \rangle=-\langle G_5^+ \rangle$,
\begin{equation}
\ket{0}\bra{0}\langle G_5^+ \rangle-\ket{1}\bra{1}\langle G_5^- \rangle=\ket{0}\bra{0}\langle G_5^+ \rangle+\ket{1}\bra{1}\langle G_5^+\rangle =\mathbbm{1}\langle G_5^+\rangle= \langle \underset{m}{\underbrace{X\dots X}}\underset{N-m-1}{\underbrace{Z\dots Z}}\mathbbm1\rangle.
\end{equation}
This completes the proof of part 4 and entire lemma.
 
\end{proof}

\section{Appendix E: 
Bell inequality violations for
fully-connected four-uniform hypergraph states}

\begin{lemma}\label{theorem4unifviol}
An arbitrary four-uniform HG state violates the classical bound  by an amount that grows exponentially along with number of qubits.
\end{lemma}
\begin{proof}
 To show this,  either  in Eq.~(\ref{bell2}) or in the original Mermin operator, $\langle \BB_N^M \rangle $,  we fix $A=Z$ and $B=X$. The choice of the Bell
 operator depends on the number of qubits: From Lemma~(\ref{lemma4full}) 
 for a given $N$ either the correlations for an even $m$ or an odd $m$ are
 given. If $m$ is even,  we choose Eq.~(\ref{bell2})  and 
 $\langle \BB_N^M \rangle $, otherwise.

 From Lemma \ref{lemma4full}, it is clear that we need to consider separate cases. However, we choose the one which encounters the smallest growth in violation and this is the $N=8k+3$ case. Other cases just encounter different factors in the growth or are greater.  For the $N=8k+3$, the strategy consists of measuring
 the Bell operator from Eq.~(\ref{bell2}) on $M=N-1$ qubits. Then we have: 
\begin{align}
\begin{split}
\langle \BB_N \rangle _Q & \geq\sum_{m=2,4\dots}^{M}\binom{M}{m}\Big(\frac{1}{\sqrt{2}}\Big)^{M-m+2} =  \frac{1}{2}\bigg[\sum_{m\ even}^{M}  \binom{M}{m} \Big(\frac{1}{\sqrt{2}}\Big)^{M-m}\bigg]-\Big(\frac{1}{\sqrt{2}}\Big)^{M+2}\\
& =\frac{1}{4}\bigg[ \Big(1+\frac{1}{\sqrt{2}}  \Big)^M  -\Big(1-\frac{1}{\sqrt{2}} \Big)^M\bigg]-\Big(\frac{1}{\sqrt{2}}\Big)^{M+2}.
\end{split}
\end{align}
Checking the ratio of the quantum  and classical values, we have that
\begin{align}
\begin{split}
\frac{\langle \BB_N \rangle _Q}{\langle \BB_N \rangle _C}
\stackrel{N\rightarrow \infty}{\sim}
\frac{\frac{1}{4}\Big(1+\frac{1}{\sqrt{2}}  \Big)^{N-1} }{2^{\frac{N-1}{2}}}
=  \frac{\Big(1+\frac{1}{\sqrt{2}}\Big)^{N-1} }{\sqrt{2}^{N+3}}
\approx  \frac{1.20711^{N}}{2\sqrt{2} +2}
.
\end{split}
\end{align}
Looking at the expectation values from Lemma \ref{lemma4full},  it is straightforward to see that in all other cases of $N$, correlations are 
stronger than in the $N=8k+3$ case, so the quantum violation increases.
\end{proof}

%----============================================================================================
\section{Appendix F: Bell and separability inequality violations for
fully-connected four-uniform hypergraph states after loosing one qubit}
\begin{lemma}
The following statement holds for $N=8k+4$ qubit, four-uniform complete hypergraph states:
\begin{equation}\label{reducedlemma}
\langle G_6 \rangle = \langle \underset{m}{\underbrace{X\dots X}}\underset{N-m-1}{\underbrace{Z\dots Z}}\mathbbm1\rangle=\begin{cases}
\begin{array}[t]{cc}
-\Big(\frac{1}{\sqrt{2}}\Big)^{N-m+2} & \mbox{if \ensuremath{m=0\  mod\  4 }},\\
+\Big(\frac{1}{\sqrt{2}}\Big)^{N-m+2} & \mbox{if \ensuremath{m=2 \ mod\  4 }}.
\end{array}\end{cases}
\end{equation}
\end{lemma}
\begin{proof}
The derivation of this result is very similar to the combinatorial calculations in the appendices D and E. Since $m$ is even, we refer to the Table 6 to see what gates remain after regrouping hyperedges on the right hand side of the expression (\ref{reducedlemma}):

\textbf{Case \# 1:} If  $m=0$ mod  4:\\
\begin{align}\label{part4_1}
\begin{split}
\langle G_6 \rangle  & =\pm\bra{+}^{\otimes N}   \underset{m}{\underbrace{X\dots X}}\underset{N-m-1}{\underbrace{Z\dots Z}}\mathbbm1 \prod_{\forall \circledast,\bigtriangleup,\diamondsuit } C_{\circledast \bigtriangleup \bigtriangleup}C_{\circledast\bigtriangleup} C_{\circledast \bigtriangleup \diamondsuit}C_{\circledast\diamondsuit}  C_{\circledast \circledast \circledast } C_{\circledast \circledast }C_{\circledast} \ket{+}^{\otimes N}\\
& =\pm \frac{1}{2^N} Tr\Big[\prod_{\forall \circledast,\bigtriangleup,\diamondsuit }  C_{\circledast \bigtriangleup \bigtriangleup}C_{\circledast\bigtriangleup} C_{\circledast \bigtriangleup \diamondsuit}C_{\circledast\diamondsuit}  C_{\circledast \circledast \circledast }C_{\circledast \circledast }C_{\circledast}C_{\bigtriangleup}  \Big].
\end{split}
\end{align}

Here $\circledast$ again refers to X operator, $\bigtriangleup$ to Z and $\diamondsuit$ to  $\mathbbm1$ and  it is denoted as $\gamma$. The strategy now is similar to the previous proofs: count the number  of +1's and -1's on the diagonal and then their difference divided by $2^N$ gives the trace.

We use $(-1)^s$ to define the sign of the diagonal element and $s=\binom{\alpha}{1}\binom{\beta}{2}+\binom{\alpha}{1}\binom{\beta}{1}+\binom{\alpha}{1}\binom{\beta}{1}\binom{\gamma}{1}+\binom{\alpha}{1}\binom{\gamma}{1}+\binom{\alpha}{3}+\binom{\alpha}{2}+\binom{\alpha}{1}+\binom{\beta}{1}$.  If $s$ is even, the value on the diagonal is $+1$ and $-1$, otherwise. We consider all possible values of $\alpha$,  $\beta$, and $\gamma$:\\

a) If $\gamma$ is even (that is $\gamma=0$):\\

\begin{tabular}{lllc}
\textbf{1. } & $\alpha$ is even \& $\beta$ is even & \multicolumn{2}{l}{if $\alpha=0$ mod 4 $\Rightarrow$ $(-1)^{s}=+1$}\tabularnewline
 &  & \multicolumn{2}{l}{if $\alpha=2$ mod 4 $\Rightarrow$ $(-1)^{s}=-1$}\tabularnewline
 &  &  & \tabularnewline
\textbf{2. } & $\alpha$ is even \& $\beta$ is odd:  & \multicolumn{2}{l}{if $\alpha=0$ mod 4 $\Rightarrow$ $(-1)^{s}=-1$}\tabularnewline
 &  & \multicolumn{2}{l}{if $\alpha=2$ mod 4 $\Rightarrow$ $(-1)^{s}=+1$}\tabularnewline
\end{tabular}\\

From here one can easily spot that for even $\alpha$, there is equal number of $+1$ and $-	1$ on the diagonal. So, they do not contribute in the calculations. We now consider the odd $\alpha$:\\

\begin{tabular}{lllc}
\textbf{3. } & $\alpha$ is odd \& $\beta$ is even & \multicolumn{2}{l}{if $\beta=0$ mod 4 $\Rightarrow$ $(-1)^{s}=-1$}\tabularnewline
 &  & \multicolumn{2}{l}{if $\beta=2$ mod 4 $\Rightarrow$ $(-1)^{s}=+1$}\tabularnewline
 &  &  & \tabularnewline
\textbf{4. } & $\alpha$ is odd \& $\beta$ is odd:  & \multicolumn{2}{l}{if $(\beta-1)=0$ mod 4 $\Rightarrow$ $(-1)^{s}=-1$}\tabularnewline
 &  & \multicolumn{2}{l}{if $(\beta-1)=2$ mod 4 $\Rightarrow$ $(-1)^{s}=+1$}\tabularnewline
\end{tabular}\\

b) If $\gamma$ is odd (that is $\gamma=1$):\\

The cases \textbf{1} and  \textbf{2} don't change. Therefore, they sum up to 0. \\

\begin{tabular}{lllc}
\textbf{3. } & $\alpha$ is odd \& $\beta$ is even & \multicolumn{2}{l}{if $\beta=0$ mod 4 $\Rightarrow$ $(-1)^{s}=+1$}\tabularnewline
 &  & \multicolumn{2}{l}{if $\beta=2$ mod 4 $\Rightarrow$ $(-1)^{s}=-1$}\tabularnewline
 &  &  & \tabularnewline
\end{tabular}\\

\textbf{4. } Stays the same as in the previous case.

Therefore the result is:
\begin{align}\label{F_1}
\begin{split}
\langle G_{6}\rangle & = \pm\frac{1}{2^{N}}\Big(\sum_{\gamma\; even}\binom{\gamma}{0}\Bigg[\sum_{\alpha\ odd}^{m}\binom{m}{\alpha}\bigg[\sum_{\beta=0,4\dots}^{N-m-1}-\binom{N-m-1}{\beta}-\binom{N-m-1}{\beta+1}+\binom{N-m-1}{\beta+2}+\binom{N-m-1}{\beta+3}\bigg]\Bigg]\\
&+\sum_{\gamma\; odd}\binom{\gamma}{1}\Bigg[\sum_{\alpha\ odd}^{m}\binom{m}{\alpha}\bigg[\sum_{\beta=0,4\dots}^{N-m-1}\binom{N-m-1}{\beta}-\binom{N-m-1}{\beta+1}-\binom{N-m-1}{\beta+2}+\binom{N-m-1}{\beta+3}\bigg]\Bigg]\Bigg)\\
&=\pm \frac{2^m}{2^N}\cdot (-Im[(1+i)^{N-m-1}])= \mp \bigg(\frac{1}{\sqrt{2}}\bigg)^{N-m+2}
\end{split}
\end{align}

It is time to fix a sign. One needs to keep in mind that the sign of  the Eq. (\ref{F_1}) is negative: if $m=4$ mod $8$, the global sign from the Table 6 is negative, and  $Im[(1+i)^{N-m-1}]=-2^\frac{N-m-2}{2}$. Therefore, an overall sign in negative. If $m=0$ mod $8$,  global  sign is positive and  $Im[(1+i)^{N-m-1}]=2^\frac{N-m-2}{2}$. Therefore, an overall sign in negative. 

%----------------------------------------------------------------------------------------
\pagebreak
\textbf{Case \# 2:} If  $m=2$ mod  4:\\
\begin{align}\label{part2ii}
\begin{split}
\langle G_6\rangle & =\pm\bra{+}^{\otimes N}  \underset{m}{\underbrace{X\dots X}}\underset{N-m-1}{\underbrace{Z\dots Z}}\mathbbm1 \prod_{\forall \circledast,\bigtriangleup,\diamondsuit } C_{\circledast \bigtriangleup \bigtriangleup} C_{\circledast \circledast \circledast } C_{\bigtriangleup \bigtriangleup } C_{\circledast \bigtriangleup \diamondsuit}  C_{\bigtriangleup \diamondsuit }\ket{+}^{\otimes N}\\
& =\pm \frac{1}{2^N} Tr( \prod_{\forall \circledast,\bigtriangleup,\diamondsuit } C_{\circledast \bigtriangleup \bigtriangleup} C_{\circledast \circledast \circledast } C_{\bigtriangleup \bigtriangleup } C_{\circledast \bigtriangleup \diamondsuit}  C_{\bigtriangleup \diamondsuit } C_{ \bigtriangleup }   ).
\end{split}
\end{align}\\

We use $(-1)^s$ to define the sign of the diagonal element and $s=\binom{\alpha}{1}\binom{\beta}{2}+\binom{\alpha}{1}\binom{\beta}{1}\binom{\gamma}{1}+\binom{\alpha}{3}+\binom{\beta}{2}+\binom{\beta}{1}\binom{\gamma}{1}+\binom{\beta}{1}$.  If $s$ is even, the value on the diagonal is $+1$ and $-1$, otherwise. We consider all possible values of $\alpha$, $\beta$, and $\gamma$:\\

a) If $\gamma$ is even  (that is $\gamma=0$):

Considering the terms from even $\alpha$:\\

\begin{tabular}{lllc}
\textbf{1. } & $\alpha$ is even \& $\beta$ is even & \multicolumn{2}{l}{if $\beta=0$ mod 4 $\Rightarrow$ $(-1)^{s}=+1$}\tabularnewline
 &  & \multicolumn{2}{l}{if $\beta=2$ mod 4 $\Rightarrow$ $(-1)^{s}=-1$}\tabularnewline
 &  &  & \tabularnewline
\textbf{2. } & $\alpha$ is even \& $\beta$ is odd:  & \multicolumn{2}{l}{if $(\beta-1)=0$ mod 4 $\Rightarrow$ $(-1)^{s}=-1$}\tabularnewline
 &  & \multicolumn{2}{l}{if $(\beta-1)=2$ mod 4 $\Rightarrow$ $(-1)^{s}=+1$}\tabularnewline
\end{tabular}

Considering odd $\alpha$:

\begin{tabular}{lllc}
\textbf{3. } & $\alpha$ is odd \& $\beta$ is even & \multicolumn{2}{l}{if $(\alpha-1)=0$ mod 4 $\Rightarrow$ $(-1)^{s}=+1$}\tabularnewline
 &  & \multicolumn{2}{l}{if $(\alpha-1)=2$ mod 4 $\Rightarrow$ $(-1)^{s}=-1$}\tabularnewline
 &  &  & \tabularnewline
\textbf{4. } & $\alpha$ is odd \& $\beta$ is odd:  & \multicolumn{2}{l}{if $(\alpha-1)=0$ mod 4 $\Rightarrow$ $(-1)^{s}=-1$}\tabularnewline
 &  & \multicolumn{2}{l}{if $(\alpha-1)=2$ mod 4 $\Rightarrow$ $(-1)^{s}=+1$}\tabularnewline
\end{tabular}\\
It is easy to see that cases \textbf{3} and \textbf{4} adds up to $0$. \\

b) If $\gamma$ is odd  (that is $\gamma=1$):\\

\textbf{1.}  Stays the same as in the previous case. \\

\textbf{2.} Gets an opposite sign, therefore, will cancel out with the a) case \textbf{2} in the sum.  \\

\textbf{3. } Stays the same as in the previous case.\\

\textbf{4. } Stays the same as in the previous case.

Therefore the result is:
\begin{align}\label{part5_2}
\begin{split}
\langle G_{6}\rangle & = \pm\frac{1}{2^{N}} 2^m\bigg[\sum_{\beta=0,4..}^{N-m-1}\binom{N-m-1}{\beta}-\binom{N-m-1}{\beta+2}\bigg]=\pm\frac{2^m}{2^N}Re[(1+i)^{N-m-1}].
\end{split}
\end{align}

It is time to fix a sign: if $m=2$ mod $8$, the global sign is positive from the Table 6, and  $Re[(1+i)^{N-m-1}]=2^\frac{N-m-2}{2}$. Therefore, an overall sign is positive. If $m=6$ mod $8$, overall sign in negative and  $Re[(1+i)^{N-m-1}]=-2^\frac{N-m-2}{2}$. Therefore, an overall sign in positive.
\end{proof} 

\begin{lemma}\label{theorem4unifviol}
A $N$-qubit ($N=8k+4$) four-uniform HG state even after tracing out one party violates the classical bound  by an amount that grows exponentially  with number of qubits. Moreover, the violation only decreases with the constant factor. 
\end{lemma}
\begin{proof}

Denote $M\equiv N-1$. Then
\begin{align}
\begin{split}
\langle \BB_N \rangle _Q & =\sum_{m=2,4\dots}^{M}\binom{M}{m}\Big(\frac{1}{\sqrt{2}}\Big)^{M-m+3}=  \frac{1}{2\sqrt{2}}\bigg[\sum_{m\ even}^{M}  \binom{M}{m} \Big(\frac{1}{\sqrt{2}}\Big)^{M-m}\bigg]-\Big(\frac{1}{\sqrt{2}}\Big)^{M+3}\\
& =\frac{1}{4\sqrt{2}}\bigg[ \Big(1+\frac{1}{\sqrt{2}}  \Big)^M  -\Big(1-\frac{1}{\sqrt{2}} \Big)^M\bigg]-\Big(\frac{1}{\sqrt{2}}\Big)^{M+3}
\end{split}
\end{align}

Checking the ratio of the quantum  and classical values, we have that
\begin{align}
\begin{split}
\frac{\langle \BB_{N-1} \rangle _Q}{\langle \BB_{N-1} \rangle _C}
\stackrel{N\rightarrow \infty}{\sim}
\frac{\frac{1}{4\sqrt{2}}\Big(1+\frac{1}{\sqrt{2}}  \Big)^{N-1} }{2^{\frac{N-2}{2}}}
=  \frac{\Big(1+\frac{1}{\sqrt{2}}\Big)^{N-1} }{\sqrt{2}^{N+3}}
\approx  \frac{1.20711^{N}}{2\sqrt{2} +2}
.
\end{split}
\end{align}
While the same value for the $N=8k+4$ qubit four-uniform complete HG state is $\frac{\langle \BB_{N} \rangle _Q}{\langle \BB_{N} \rangle _C}
\stackrel{N\rightarrow \infty}{\sim}\approx\frac{1.20711^{N}}{4}.$ Therefore, after tracing out a single qubit, the local realism violation decreases with the small constant factor.
\end{proof}

It is important to note that the similar violation is maintained after tracing out more than one qubit. For example, numerical evidence confirms for $N=12$, that if one takes a Bell inequality with the odd number of $X$ measurements, instead of the even ones as we have chosen in the proof, an exponential violation is maintained after tracing out 2 qubits. But even if 5 qubit is traced out, the state is stilled entangled and this can be verified using the separability inequality \cite{Roy}.
Exact violation is given in the following table:  

\begin{center}
\begin{tabular}{|c|c|c|c|c|}

\hline 
\#k & Quantum Value  & Classical Bound & Separability Bound & $\approx$Ratio\tabularnewline
\hline 
\hline 
\textbf{0} & \textbf{153.141} & \textbf{64} &  & \textbf{2.39283}\tabularnewline
\hline 
\textcolor{red}{$1$} & \textcolor{red}{89.7188} & \textcolor{red}{32} &  & \textcolor{red}{2.78125}\tabularnewline
\hline 
$2$ & 37.1563 & 32 &  & 1.16113\tabularnewline
\hline 
\multicolumn{1}{c}{} & \multicolumn{1}{c}{} & \multicolumn{1}{c}{} & \multicolumn{1}{c}{} & \multicolumn{1}{c}{}\tabularnewline
\hline 
$3$ & 15.4219 & - & $\sqrt{2}$ & 10.9049\tabularnewline
\hline 
$4$ & 6.375 & - & $\sqrt{2}$ & 4.50781\tabularnewline
\hline 
$5$ & 2.70313 & - & $\sqrt{2}$ & 1.9114\tabularnewline
\hline 
\end{tabular}\\
$\quad$\\
\textbf{Table 7.} Violation of Bell (for odd $m$) and Separability inequalities in $N=12$ qubit 4-uniform HG state. Here k is the number of traced out qubits. Red line represents that when k = 1, or equivalently one qubit is traced out, the violation of Bell inequalities increases. This is caused by decrease in the classical bound \cite{Mermin}. 
\end{center}

%===========================================================================================

% APPENDIX G

%===========================================================================================

\section{Appendix G: Separability  inequality violations for
fully-connected three-uniform hypergraph states after loosing one qubit}
\begin{lemma}
The following statements holds three-uniform complete hypergraph states:\\
(i) For $N=8k+5$, $N=8k+6$, or $N=8k+7$: 
\begin{equation}\label{G_1}
\langle G_7 \rangle =\langle\underset{m}{\underbrace{X\dots X}}\underset{N-m-1}{\underbrace{Z\dots Z}}\mathbbm1\rangle=\begin{cases}
\begin{array}[t]{cc}
-\frac{1}{2^{\left\lfloor\frac{N-1}{2}\right\rfloor}} & \mbox{if \ensuremath{(m-1)=0\  mod\  4 }},\\
+\frac{1}{2^{\left\lfloor\frac{N-1}{2}\right\rfloor}} & \mbox{if \ensuremath{(m-1)=2 \ mod\  4 }}.
\end{array}\end{cases}
\end{equation}
(ii) For $N=8k+1$, $N=8k+2$, or $N=8k+3$:
\begin{equation}\label{G_2}
\langle G_7 \rangle =\langle\underset{m}{\underbrace{X\dots X}}\underset{N-m-1}{\underbrace{Z\dots Z}}\mathbbm1\rangle=\begin{cases}
\begin{array}[t]{cc}
+\frac{1}{2^{\left\lfloor\frac{N-1}{2}\right\rfloor}} & \mbox{if \ensuremath{(m-1)=0\  mod\  4 }},\\
-\frac{1}{2^{\left\lfloor\frac{N-1}{2}\right\rfloor}} & \mbox{if \ensuremath{(m-1)=2 \ mod\  4 }}.
\end{array}\end{cases}
\end{equation}
(iii) For $N=4k$:
\begin{equation}\label{G_3}
\langle G_7 \rangle =\langle\underset{m}{\underbrace{X\dots X}}\underset{N-m-1}{\underbrace{Z\dots Z}}\mathbbm1\rangle=0
\end{equation}
\end{lemma}

\begin{proof}
As a starting point, we derive the remaining gates on the right hand side of the equations (\ref{G_1}), the same derivation turns out to be working for Eq. (\ref{G_2}) and (\ref{G_3}). This approach is analogous to the Appendix B, but now, $m$ is odd.

\begin{tabular}{lrl}
 &  & \tabularnewline
$1.$ & $\# \circledast\circledast=$ & $\binom{m-2}{1}$ is odd $\Rightarrow \;C_{\circledast\circledast}$ remains.\tabularnewline
 &  & \tabularnewline
$2.$ & $\#\circledast=$ & $\binom{m-1}{2}=\frac{(m-1)(m-2)}{2}$ is $\begin{cases}
\begin{array}[t]{cc}
\mbox{even,  if (\ensuremath{m-1)=0}  mod 4 } &\Rightarrow C_{\circledast}\mbox{ cancels.}\\
\mbox{odd,  if (\ensuremath{m-1)=2}  mod 4 } & \Rightarrow C_{\circledast}\mbox{ remains.}
\end{array}\end{cases}$ \tabularnewline
 &  & \tabularnewline
$3.$  & $ \#\{\}=$ & $\binom{m}{3}=\frac{m(m-1)(m-2)}{2\cdot3}$ is $\begin{cases}
\begin{array}[t]{cc}
\mbox{even,  if (\ensuremath{m-1)=0}  mod 4 } &\Rightarrow\mbox{ gives a positive sign.}\\
\mbox{odd,  if (\ensuremath{m-1)=2}  mod 4 } & \Rightarrow\mbox{ gives a negative sign.}
\end{array}\end{cases}$ \tabularnewline
 &  & \tabularnewline
$4.$ & $\#\bigtriangleup\bigtriangleup=$ &  $\binom{m}{1}$ is odd $\Rightarrow \; C_{\bigtriangleup\bigtriangleup}$ remains.  \tabularnewline
$5.$ & $\#\circledast\bigtriangleup=$ &  $\binom{m-1}{1}$ is even $\Rightarrow \; C_{\circledast\bigtriangleup}$ cancels.  \tabularnewline
$6.$  & $ \# \bigtriangleup=$ & $\binom{m}{2}=\frac{m(m-1)}{2\cdot3}$ is $\begin{cases}
\begin{array}[t]{cc}
\mbox{even,  if (\ensuremath{m-1)=0}  mod 4 } &\Rightarrow C_{\bigtriangleup}\mbox{ cancels.}\\
\mbox{odd,  if (\ensuremath{m-1)=2}  mod 4 } & \Rightarrow C_{\bigtriangleup}\mbox{ remains.}
\end{array}\end{cases}$ \tabularnewline
 &  & \tabularnewline

\end{tabular}
\begin{center}
\textbf{Table 8.} Counting phase gates for a three-uniform HG when $m$ ($m$ is odd) systems are measured in $X$ direction.
\end{center}

Consider two cases:\\

\textbf{1.} If $(m-1)=0$ mod $4$:
\begin{align}\label{part4_1}
\begin{split}
\langle G_7 \rangle  & =\pm\bra{+}^{\otimes N}   \underset{m}{\underbrace{X\dots X}}\underset{N-m-1}{\underbrace{Z\dots Z}}\mathbbm1 \prod_{\forall \circledast,\bigtriangleup,\diamondsuit } C_{\circledast\circledast}C_{\bigtriangleup\bigtriangleup} C_{\bigtriangleup\diamondsuit}   \ket{+}^{\otimes N}\\
& =\pm \frac{1}{2^N} Tr\Big[C_{\circledast\circledast}C_{\bigtriangleup\bigtriangleup} C_{\bigtriangleup\diamondsuit}  C_{\bigtriangleup} \Big].
\end{split}
\end{align}

Here $\circledast$ again refers to X operator, $\bigtriangleup$ to Z and $\diamondsuit$ to  $\mathbbm1$ and is denoted by $\gamma$. The strategy is similar to the previous case: count the number  of +1's and -1's on the diagonal. Their difference divided by $2^N$, gives the trace. \\

We use $(-1)^s$ to define the sign of the diagonal element and $s=\binom{\alpha}{2}+\binom{\beta}{2}+\binom{\beta}{1}\binom{\gamma}{1}+\binom{\beta}{1}$.  If $s$ is even, the value on the diagonal is $+1$ and $-1$, otherwise. We consider all possible values of $\alpha$, $\beta$  and $\gamma$:\\

a) If $\gamma$ is even (that is $\gamma=0$ ):\\

Considering the terms from even $\beta$:\\

\begin{tabular}{lllc}
\textbf{1. } & $\beta$ is even \& $\alpha$ is even & \multicolumn{2}{l}{if $\alpha=0$ mod 4 and  $\beta=0$ mod 4  $\Rightarrow$ $(-1)^{s}=+1$}\tabularnewline
 &  & \multicolumn{2}{l}{if $\alpha=0$ mod 4 and  $\beta=2$ mod 4  $\Rightarrow$ $(-1)^{s}=-1$}\tabularnewline
 &  & \multicolumn{2}{l}{if $\alpha=2$ mod 4 and  $\beta=0$ mod 4  $\Rightarrow$ $(-1)^{s}=-1$}\tabularnewline
 &  & \multicolumn{2}{l}{if $\alpha=2$ mod 4 and  $\beta=2$ mod 4  $\Rightarrow$ $(-1)^{s}=+1$}\tabularnewline
 &  &  & \tabularnewline
\textbf{2. } & $\beta$ is even \& $\alpha$ is odd & \multicolumn{2}{l}{if $(\alpha-1)=0$ mod 4 and  $\beta=0$ mod 4  $\Rightarrow$ $(-1)^{s}=+1$}\tabularnewline
 &  & \multicolumn{2}{l}{if $(\alpha-1)=0$ mod 4 and  $\beta=2$ mod 4  $\Rightarrow$ $(-1)^{s}=-1$}\tabularnewline
 &  & \multicolumn{2}{l}{if $(\alpha-1)=2$ mod 4 and  $\beta=0$ mod 4  $\Rightarrow$ $(-1)^{s}=-1$}\tabularnewline
 &  & \multicolumn{2}{l}{if $(\alpha-1)=2$ mod 4 and  $\beta=2$ mod 4  $\Rightarrow$ $(-1)^{s}=+1$}\tabularnewline
 &  &  & \tabularnewline
\end{tabular}\\

Considering odd $\beta$:\\

\begin{tabular}{lllc}
\textbf{3. } & $\beta$ is odd \& $\alpha$ is even & \multicolumn{2}{l}{if $\alpha=0$ mod 4 and  $(\beta-1)=0$ mod 4  $\Rightarrow$ $(-1)^{s}=-1$}\tabularnewline
 &  & \multicolumn{2}{l}{if $\alpha=0$ mod 4 and  $(\beta-1)=2$ mod 4  $\Rightarrow$ $(-1)^{s}=+1$}\tabularnewline
 &  & \multicolumn{2}{l}{if $\alpha=2$ mod 4 and  $(\beta-1)=0$ mod 4  $\Rightarrow$ $(-1)^{s}=+1$}\tabularnewline
 &  & \multicolumn{2}{l}{if $\alpha=2$ mod 4 and  $(\beta-1)=2$ mod 4  $\Rightarrow$ $(-1)^{s}=-1$}\tabularnewline
 &  &  & \tabularnewline
\textbf{4. } & $\beta$ is odd \& $\alpha$ is odd & \multicolumn{2}{l}{if $(\alpha-1)=0$ mod 4 and  $(\beta-1)=0$ mod 4  $\Rightarrow$ $(-1)^{s}=-1$}\tabularnewline
 &  & \multicolumn{2}{l}{if $(\alpha-1)=0$ mod 4 and  $(\beta-1)=2$ mod 4  $\Rightarrow$ $(-1)^{s}=+1$}\tabularnewline
 &  & \multicolumn{2}{l}{if $(\alpha-1)=2$ mod 4 and  $(\beta-1)=0$ mod 4  $\Rightarrow$ $(-1)^{s}=+1$}\tabularnewline
 &  & \multicolumn{2}{l}{if $(\alpha-1)=2$ mod 4 and  $(\beta-1)=2$ mod 4  $\Rightarrow$ $(-1)^{s}=-1$}\tabularnewline
 &  &  & \tabularnewline
\end{tabular}

b) If $\gamma$ is odd (that is $\gamma=1$):\\

Considering the terms from even $\beta$:\\

\textbf{1.} and  \textbf{2.}  Nothing  changes  in comparison  to a) \textbf{1.} and \textbf{2.}\\

\textbf{3.} and  \textbf{4.}  These two terms have opposite sign from a) \textbf{3.} and \textbf{4.} Therefore, in the sum they cancel ($\gamma$ is always  1.)\\

Therefore,
\begin{align}
\begin{split}
\langle G_{7}\rangle & = \frac{1}{2^{N-1}}\bigg[\sum_{\beta=0,4..}^{N-m-1}\binom{N-m-1}{\beta}-\binom{N-m-1}{\beta+2}\bigg]\bigg[\sum_{\alpha=0,4..}^{m}\binom{m}{\alpha}+\binom{m}{\alpha+1}-\binom{m}{\alpha+2}-\binom{m}{\alpha+3}\bigg]\\
&= \frac{1}{2^{N-1}}Re[(1+i)^{N-m-1}]\cdot\Big(Re[(1+i)^m]+Im[(1+i)^m]\Big).
\end{split}
\end{align}

Now we can consider each cases separately, for this we use the lookup \textbf{table 0}:\\

$(i)$ If $N=8k+5$ and if $(m-1)=4$ mod $8$,  $Re[(1+i)^{N-m-1}]=-2^\frac{N-m-2}{2}$ and $Re[(1+i)^m]+Im[(1+i)^m]=2^\frac{m+1}{2}$. Therefore, 
\begin{equation}
\langle G_{7}\rangle =-\frac{2^\frac{N-m-2}{2}\cdot 2^\frac{m+1}{2} }{2^{N-1}}=-\Big(\frac{1}{2}\Big)^{\frac{N-1}{2}}.
\end{equation}
And, if  $(m-1)=0$ mod $8$,  $Re[(1+i)^{N-m-1}]=+2^\frac{N-m-2}{2}$ and $Re[(1+i)^m]+Im[(1+i)^m]=-2^\frac{m+1}{2}$. Therefore, 
\begin{equation}
\langle G_{7}\rangle =-\frac{2^\frac{N-m-2}{2}\cdot 2^\frac{m+1}{2} }{2^{N-1}}=-\Big(\frac{1}{2}\Big)^{\frac{N-1}{2}}.
\end{equation}
Exactly the same is true for  $N=8k+7$. But for $N=8k+6$ and if $(m-1)=4$ mod $8$,  $Re[(1+i)^{N-m-1}]=-2^\frac{N-m-1}{2}$ and $Re[(1+i)^m]+Im[(1+i)^m]=2^\frac{m+1}{2}$. Therefore, 
\begin{equation}
\langle G_{7}\rangle =-\frac{2^\frac{N-m-1}{2}\cdot 2^\frac{m+1}{2} }{2^{N-1}}=-\Big(\frac{1}{2}\Big)^{\left\lfloor\frac{N-1}{2}\right\rfloor}.
\end{equation}
And,  if $(m-1)=0$ mod $8$,  $Re[(1+i)^{N-m-1}]=2^\frac{N-m-1}{2}$ and $Re[(1+i)^m]+Im[(1+i)^m]=-2^\frac{m+1}{2}$. Therefore, 
\begin{equation}
\langle G_{7}\rangle =-\frac{2^\frac{N-m-1}{2}\cdot 2^\frac{m+1}{2} }{2^{N-1}}=-\Big(\frac{1}{2}\Big)^{\left\lfloor\frac{N-1}{2}\right\rfloor}.
\end{equation}\\

$(ii)$   If $N=8k+1$ and if $(m-1)=4$ mod $8$,  $Re[(1+i)^{N-m-1}]=-2^\frac{N-m-2}{2}$ and $Re[(1+i)^m]+Im[(1+i)^m]=-2^\frac{m+1}{2}$. Therefore, 
\begin{equation}
\langle G_{7}\rangle =\frac{2^\frac{N-m-2}{2}\cdot 2^\frac{m+1}{2} }{2^{N-1}}=\Big(\frac{1}{2}\Big)^{\frac{N-1}{2}}.
\end{equation}
And,  if $(m-1)=0$ mod $8$,  $Re[(1+i)^{N-m-1}]=2^\frac{N-m-1}{2}$ and $Re[(1+i)^m]+Im[(1+i)^m]=2^\frac{m+1}{2}$. Therefore, 
\begin{equation}
\langle G_{7}\rangle =\frac{2^\frac{N-m-1}{2}\cdot 2^\frac{m+1}{2} }{2^{N-1}}=\Big(\frac{1}{2}\Big)^{\left\lfloor\frac{N-1}{2}\right\rfloor}.
\end{equation}
It it analogous for other two cases as well.\\

$(iii)$ For $N=4k$,   $Re[(1+i)^{N-m-1}]=0$. Therefore, $\langle G_{7}\rangle =0.$ \\

%=========================================================================================

\textbf{2.} If $(m-1)=2$ mod $4$:
\begin{align}\label{part4_1}
\begin{split}
\langle G_7 \rangle  & =\bra{+}^{\otimes N}   \underset{m}{\underbrace{X\dots X}}\underset{N-m-1}{\underbrace{Z\dots Z}}\mathbbm1 \prod_{\forall \circledast,\bigtriangleup,\diamondsuit } C_{\circledast\circledast}C_{\circledast}C_{\bigtriangleup\bigtriangleup} C_{\bigtriangleup\diamondsuit}C_{\bigtriangleup}C_{\diamondsuit}    \ket{+}^{\otimes N}\\
& = \frac{1}{2^N} Tr\Big[C_{\circledast\circledast}C_{\circledast}C_{\bigtriangleup\bigtriangleup} C_{\bigtriangleup\diamondsuit}C_{\diamondsuit}    \Big].
\end{split}
\end{align}

Here $\circledast$ again refers to X operator, $\bigtriangleup$ to Z and $\diamondsuit$ to  $\mathbbm1$ and is denoted by $\gamma$. The strategy is similar to the previous case: count the number  of +1's and -1's on the diagonal and their difference divided by $2^N$, gives the trace. \\

We use $(-1)^s$ to define the sign of the diagonal element and $s=\binom{\alpha}{2}+\binom{\alpha}{1}+\binom{\beta}{2}+\binom{\beta}{1}\binom{\gamma}{1}+\binom{\gamma}{1}$.  If $s$ is even, the value on the diagonal is $+1$ and $-1$, otherwise. We consider all possible values of $\alpha$, $\beta$ and $\gamma$:\\

a) If $\gamma$ is even (that is $\gamma=0$):\\

Considering the terms from even $\beta$:\\

\begin{tabular}{lllc}
\textbf{1. } & $\beta$ is even \& $\alpha$ is even & \multicolumn{2}{l}{if $\alpha=0$ mod 4 and  $\beta=0$ mod 4  $\Rightarrow$ $(-1)^{s}=+1$}\tabularnewline
 &  & \multicolumn{2}{l}{if $\alpha=0$ mod 4 and  $\beta=2$ mod 4  $\Rightarrow$ $(-1)^{s}=-1$}\tabularnewline
 &  & \multicolumn{2}{l}{if $\alpha=2$ mod 4 and  $\beta=0$ mod 4  $\Rightarrow$ $(-1)^{s}=-1$}\tabularnewline
 &  & \multicolumn{2}{l}{if $\alpha=2$ mod 4 and  $\beta=2$ mod 4  $\Rightarrow$ $(-1)^{s}=+1$}\tabularnewline
 &  &  & \tabularnewline
\textbf{2. } & $\beta$ is even \& $\alpha$ is odd & \multicolumn{2}{l}{if $(\alpha-1)=0$ mod 4 and  $\beta=0$ mod 4  $\Rightarrow$ $(-1)^{s}=-1$}\tabularnewline
 &  & \multicolumn{2}{l}{if $(\alpha-1)=0$ mod 4 and  $\beta=2$ mod 4  $\Rightarrow$ $(-1)^{s}=+1$}\tabularnewline
 &  & \multicolumn{2}{l}{if $(\alpha-1)=2$ mod 4 and  $\beta=0$ mod 4  $\Rightarrow$ $(-1)^{s}=+1$}\tabularnewline
 &  & \multicolumn{2}{l}{if $(\alpha-1)=2$ mod 4 and  $\beta=2$ mod 4  $\Rightarrow$ $(-1)^{s}=-1$}\tabularnewline
 &  &  & \tabularnewline
\end{tabular}\\

Considering odd $\beta$:\\

\begin{tabular}{lllc}
\textbf{3. } & $\beta$ is odd \& $\alpha$ is even & \multicolumn{2}{l}{if $\alpha=0$ mod 4 and  $(\beta-1)=0$ mod 4  $\Rightarrow$ $(-1)^{s}=+1$}\tabularnewline
 &  & \multicolumn{2}{l}{if $\alpha=0$ mod 4 and  $(\beta-1)=2$ mod 4  $\Rightarrow$ $(-1)^{s}=-1$}\tabularnewline
 &  & \multicolumn{2}{l}{if $\alpha=2$ mod 4 and  $(\beta-1)=0$ mod 4  $\Rightarrow$ $(-1)^{s}=-1$}\tabularnewline
 &  & \multicolumn{2}{l}{if $\alpha=2$ mod 4 and  $(\beta-1)=2$ mod 4  $\Rightarrow$ $(-1)^{s}=+1$}\tabularnewline
 &  &  & \tabularnewline
\textbf{4. } & $\beta$ is odd \& $\alpha$ is odd & \multicolumn{2}{l}{if $(\alpha-1)=0$ mod 4 and  $(\beta-1)=0$ mod 4  $\Rightarrow$ $(-1)^{s}=-1$}\tabularnewline
 &  & \multicolumn{2}{l}{if $(\alpha-1)=0$ mod 4 and  $(\beta-1)=2$ mod 4  $\Rightarrow$ $(-1)^{s}=+1$}\tabularnewline
 &  & \multicolumn{2}{l}{if $(\alpha-1)=2$ mod 4 and  $(\beta-1)=0$ mod 4  $\Rightarrow$ $(-1)^{s}=+1$}\tabularnewline
 &  & \multicolumn{2}{l}{if $(\alpha-1)=2$ mod 4 and  $(\beta-1)=2$ mod 4  $\Rightarrow$ $(-1)^{s}=-1$}\tabularnewline
 &  &  & \tabularnewline
\end{tabular}

b) If $\gamma$ is odd  (that is $\gamma=1$)::\\

\textbf{1.} and  \textbf{2.}  These two terms have opposite sign from a) \textbf{1.} and \textbf{2.} Therefore, in the sum they cancel ($\gamma$ is always 1.)\\

\textbf{3.} and  \textbf{4.}  Nothing  changes  in comparison  to a) \textbf{3.} and \textbf{4.}\\

Therefore,
\begin{align}
\begin{split}
\langle G_{7}\rangle & = \frac{1}{2^{N-1}}\bigg[\sum_{\beta=0,4..}^{N-m-1}\binom{N-m-1}{\beta+1}-\binom{N-m-1}{\beta+3}\bigg]\bigg[\sum_{\alpha=0,4..}^{m}\binom{m}{\alpha}-\binom{m}{\alpha+1}-\binom{m}{\alpha+2}+\binom{m}{\alpha+3}\bigg]\\
&= \frac{1}{2^{N-1}}Im[(1+i)^{N-m-1}]\cdot\Big(Re[(1+i)^m]-Im[(1+i)^m]\Big).
\end{split}
\end{align}
Now we can consider each cases separately, for this we use the lookup \textbf{table 0}:\\

$(i)$ If $N=8k+5$ and if $(m-1)=4$ mod $8$,  $Im[(1+i)^{N-m-1}]=2^\frac{N-m-2}{2}$ and $Re[(1+i)^m]-Im[(1+i)^m]=-2^\frac{m+1}{2}$. And the overall sign from the Table 8 is negative.  Therefore, 
\begin{equation}
\langle G_{7}\rangle =-\frac{2^\frac{N-m-2}{2}\cdot (-2^\frac{m+1}{2}) }{2^{N-1}}=\Big(\frac{1}{2}\Big)^{\frac{N-1}{2}}.
\end{equation}
And, if  $(m-1)=0$ mod $8$,  $Im[(1+i)^{N-m-1}]=-2^\frac{N-m-2}{2}$ and $Re[(1+i)^m]-Im[(1+i)^m]=2^\frac{m+1}{2}$. Overall sign in negative. Therefore, 
\begin{equation}
\langle G_{7}\rangle =-\frac{-2^\frac{N-m-2}{2}\cdot 2^\frac{m+1}{2} }{2^{N-1}}=\Big(\frac{1}{2}\Big)^{\frac{N-1}{2}}.
\end{equation}
The same holds for  $N=8k+6$ and $N=8k+7$.  Besides,  $(ii)$ differs with the sign flip and is trivial to check.It is also trivial to prove $(iii)$,  as when $N=4k$, $Im[(1+i)^{N-m-1}]=0$.
\end{proof}

\begin{lemma}\label{theorem3unifviol}
A $N$-qubit (except when $N=4k$) three-uniform HG state violates the separability inequality exponentially after tracing out a single qubit.
\end{lemma}
\begin{proof}
We consider only one case, $N=8k+5$, as others are analogous. Here $M=N-1$:\\
\begin{align}
\begin{split}
\langle \BB_N \rangle _Q & =\sum_{m \; odd}^{M}\binom{M}{m}\Big(\frac{1}{\sqrt{2}}\Big)^{M}=  2^{M-1}\cdot2^{-M/2}=\sqrt{2}^{N-3}. 
\end{split}
\end{align}
Separability bound is $\sqrt{2}$ \cite{Roy} and it does not depend on the number of qubits.
\end{proof}

It is important to note that the similar violation is maintained after tracing out more than one qubit. Numerical evidence  for $N=11$ is presented below.\\

\begin{center}
\begin{tabular}{|c|c|c|c|}
\hline 
\#k & Quantum Value  & Separability Bound & $\approx$Ratio\tabularnewline
\hline 
\hline 
\textbf{0} & 511.5 & $\sqrt{2}$ & 361.69\tabularnewline
\hline 
{$1$} & 16 & $\sqrt{2}$ & 11.31\tabularnewline
\hline 
$2$ & 8 & $\sqrt{2}$ & 5.66\tabularnewline
\hline 
$3$ & 4 & $\sqrt{2}$ & 2.83\tabularnewline
\hline 
$4$ & 2 & $\sqrt{2}$ & 1.414\tabularnewline
\hline 
\end{tabular}\\
$\quad$\\
\textbf{Table 9.} Violation of the separability inequalities \cite{Roy} in $N=11$-qubit 3-uniform HG states. Here $k$ is the number of traced out qubits. When $k =0$, the Mermin inequality is violated as expected.
\end{center}

\twocolumngrid


\begin{thebibliography}{99}

\bibitem{entangleduseful}
P. Hayden, D.W. Leung, and A. Winter,
Comm. Math. Phys. {\bf 265}, 95 (2006);
D. Gross, S. Flammia, and J. Eisert,
Phys. Rev. Lett. {\bf 102}, 190501 (2009);
M. J. Bremner, C. Mora, and A. Winter,
Phys. Rev. Lett. {\bf 102}, 190502 (2009).

\bibitem{hein}
M. Hein, W. D\"ur, J. Eisert, R. Raussendorf, M. Van den Nest, and H.-J. Briegel,
{\it Entanglement in Graph States and its Applications},
in {\em Quantum Computers, Algorithms and Chaos}, edited by G.
Casati, D.L. Shepelyansky, P. Zoller, and G. Benenti (IOS Press,
Amsterdam, 2006), quant-ph/0602096.


\bibitem{Mermin}
N. D. Mermin,
Phys. Rev. Lett. {\bf 65}, 1838 (1990).

\bibitem{gthb}
O. G\"uhne, G. T\'oth, P. Hyllus, and H. J. Briegel,
Phys. Rev. Lett. {\bf 95}, 120405 (2005); 
G. T\'oth, O. G\"uhne, H.~J. Briegel,
Phys. Rev. A {\bf 73}, 022303 (2006).


\bibitem{gtreview}
O. G\"uhne and G. T\'oth,
{Phys. Rep.} {\bf 474}, 1 (2009).

\bibitem{Kruszynska2009}
C. Kruszynska and B. Kraus,
Phys. Rev. A {\bf 79}, 052304 (2009).

\bibitem{Qu2013_encoding}
R.~Qu, J.~Wang, Z.~Li, and Y.~Bao,
Phys. Rev. A {\bf 87}, 022311 (2013).

\bibitem{Rossi2013}
M.~Rossi, M.~Huber, D.~Bru\ss, and C.~Macchiavello,
New J. Phys. {\bf 15}, 113022 (2013).

\bibitem{Otfried}
O. G\"uhne, M. Cuquet, F. E. S. Steinhoff, T. Moroder, M. Rossi,
D. Bru{\ss}, B. Kraus, and  C. Macchiavello,
J. Phys. A: Math. Theor. {\bf 47}, 335303 (2014).

\bibitem{chenlei}
X.-Y. Chen and L. Wang, 
J. Phys. A: Math. Theor. {\bf 47}, 415304 (2014).

\bibitem{lyons}
D. W. Lyons, D. J. Upchurch, S. N. Walck, and C. D. Yetter,
J. Phys. A: Math. Theor. {\bf 48}, 095301 (2015).

\bibitem{scripta}
M. Rossi, D. Bru{\ss}, and C. Macchiavello,
Phys. Scr. {\bf T160}, 014036 (2014). 

\bibitem{mora}
H. Buhrman, R. Cleve, J. Watrous, and R. de Wolf,
Phys. Rev. Lett. {\bf 87}, 167902 (2001),
C.E. Mora, H.J. Briegel, and B. Kraus, 
Int. J. Quantum Inform. {\bf 5}, 729 (2007). 

\bibitem{qma}
A. B. Grilo, I. Kerenidis, and J. Sikora,
arXiv:1410.2882.

\bibitem{Yoshida}
B. Yoshida, arXiv:1508.03468.

\bibitem{Akimasa}
J. Miller and A Miyake, arXiv:1508.02695.


\bibitem{Hardy92} 
L. Hardy, 
Phys. Rev. Lett. {\bf 71}, 1665 (1993).

\bibitem{Svetlichny87} 
G. Svetlichny, 
Phys. Rev. D {\bf 35}, 3066 (1987).

\bibitem{popescurohrlich}
S. Popescu and D. Rohrlich,
Phys. Lett. A {\bf 166}, 293 (1992).

\bibitem{wwzb}
R. F. Werner and M. M. Wolf, 
Phys. Rev. A {\bf 64}, 032112 (2001);
M. \.Zukowski and C. Brukner, 
Phys. Rev. Lett. {\bf 88}, 210401 (2002).


\bibitem{brunnerreview}
N. Brunner, D. Cavalcanti, S. Pironio, V. Scarani, and S. Wehner,
Rev. Mod. Phys. {\bf 86}, 419 (2014).

\bibitem{brukner}
C. Brukner, M. \.{Z}ukowski, J.-W. Pan, and A. Zeilinger,
Phys. Rev. Lett. {\bf 92}, 127901 (2004). 


\bibitem{Abramsky}
S. Abramsky and C. Constantin,
Electronic Proc. Theor. Comp. Science
{\bf 171}, 10 (2014), arXiv:1412.5213.

\bibitem{Wang}
Z. Wang and D. Markham,
Phys. Rev. Lett. {\bf 108}, 210407 (2012).

\bibitem{appremark}
The appendix can be found in the supplementary 
online material.

\bibitem{Roy} 
S. M. Roy,
Phys. Rev. Lett. {\bf 94}, 010402 (2005).

\bibitem{giovanetti}
V. Giovannetti, S. Lloyd, and L. Maccone, 
Science {\bf 306}, 1330 (2004).

\bibitem{weibo}
W.-B. Gao, C.-Y. Lu, X.-C. Yao, P. Xu, O. G\"uhne, 
A. Goebel, Y.-A. Chen, C.-Z. Peng, Z.-B. Chen, and J.-W. Pan,
Nature Phys. {\bf 6}, 331 (2010).

\bibitem{Hoban11}
 M. J. Hoban, E. T. Campbell, K. Loukopoulos, and D. E. Browne, 
 New J. Phys. {\bf 13}, 023014 (2011).

\bibitem{Collins02} 
D. Collins, N. Gisin, S. Popescu, D. Roberts, and V. Scarani, 
Phys. Rev. Lett. {\bf 88}, 170405 (2002).

\bibitem{Buzek}
M. Koashi ,V. Buzek, and N. Imoto,
Phys. Rev. A 62, 050302(R) (2000)


\end{thebibliography}
\end{document}